\documentclass[12pt]{amsart}
\usepackage{amssymb,amscd}
\usepackage[colorlinks=true,allcolors=blue]{hyperref}

\let\mc\mathcal
\usepackage{eucal}
\let\euc\mathcal
\let\mathcal\mc

\def\>{\relax\ifmmode\mskip.666667\thinmuskip\relax\else\kern.111111em\fi}
\def\:{\relax\ifmmode\mskip.333333\thinmuskip\relax\else\kern.0555556em\fi}
\def\?{\relax\ifmmode\mskip-.666667\thinmuskip\relax\else\kern-.111111em\fi}
\def\<{\relax\ifmmode\mskip-.333333\thinmuskip\relax\else\kern-.0555556em\fi}
\def\vsk#1>{\vskip#1\baselineskip}
\def\vv#1>{\vadjust{\vsk#1>}\ignorespaces}
\def\vvn#1>{\vadjust{\nobreak\vsk#1>\nobreak}\ignorespaces}
 \let\alb\allowbreak
\def\fratop{\genfrac{}{}{0pt}1}
\def\satop#1#2{\fratop{\scriptstyle#1}{\scriptstyle#2}}
\let\dsize\displaystyle \let\tsize\textstyle \let\ssize\scriptstyle
\let\sssize\scriptscriptstyle 
\let\phan\phantom \let\vp\vphantom 
\def\stackrel#1#2{\mathrel{\mathop{\kern 0pt#2}\limits^{#1}}}

\let\Medskip\medskip
\def\medskip{\par\Medskip}
\let\Bigskip\bigskip
\def\bigskip{\par\Bigskip}

\let\Maketitle\maketitle
\def\maketitle{\Maketitle\thispagestyle{empty}\let\maketitle\empty}

\newtheorem{thm}{Theorem}[section]
\newtheorem{cor}[thm]{Corollary}
\newtheorem{lem}[thm]{Lemma}
\newtheorem{prop}[thm]{Proposition}

\numberwithin{equation}{section}

\theoremstyle{definition}

\newtheorem*{example}{Example}

\def\beq{\begin{equation}}
\def\eeq{\end{equation}}
\def\be{\begin{equation*}}
\def\ee{\end{equation*}}

\let\nc\newcommand

\def\bea{\begin{eqnarray*}}
\def\eea{\end{eqnarray*}}
\def\bean{\begin{eqnarray}}
\def\eean{\end{eqnarray}}

\let\al\alpha
\let\bt\beta
\let\gm\gamma
\let\Gm\Gamma
\let\dl\delta

\let\ka\kappa
\let\la\lambda

\let\phi\varphi
\let\si\sigma

\let\der\partial
\let\Hat\widehat

\let\ox\otimes
\let\Tilde\widetilde
\let\bra\langle
\let\ket\rangle

\let\ge\geqslant

\let\le\leqslant

\let\on\operatorname
\let\bi\bibitem
\let\bs\boldsymbol

\def\C{{\mathbb C}}
\def\Z{{\mathbb Z}}
\def\R{{\mathbb R}}
\def\BB{{\mathbb B}}
\def\Ibb{{\mathbb I}}

\def\Ub{{\mathbb U}}

\def\Ie{{\euc I}}
\def\Ue{{\euc U}}
\def\We{{\euc W}}

\def\Fc{{\mc F}}
\def\Hc{{\mc H}}
\def\Ic{{\mc I}}

\def\Pc{{\mc P}}

\def\Me{{\euc M}}
\def\Te{{\euc T}}

\def\Eb{\;\overline{\!\?E\?}\>}
\def\gmb{\bs\gm}
\def\gmbr{\>\bar{\?\gm}}
\def\gmbb{\bs{\gmbr}}
\def\Gmb{\bar\Gm}
\def\GGb{\bs{\Gmb}}

\def\Ibr{\bar I}
\def\Jbr{\bar J}
\def\Kbr{\bar K}

\def\+#1{^{\{#1\}}}
\def\lsym#1{#1\alb\dots\relax#1\alb}
\def\lc{\lsym,}

\def\Gr{\mathrm{Gr}}

\def\Im{\on{Im}}
\def\Re{\on{Re}}
\def\Res{\on{Res}}

\def\tbigoplus{\mathop{\textstyle{\bigoplus}}\limits}

\def\cirs{{\raise.2ex\hbox{$\sssize\circ$}}}
\def\Loc{\on{Loc}}

\def\ab{a,\:b}

\def\ij{i,\:j}

\def\kn{k\<,\:n}

\def\IJ{I\?,\:J}

\def\JK{J,\:K}

\def\gsl{\mathfrak{sl}}

\def\onek{\{\:1\:\lc k\:\}}
\def\onen{\{\:1\:\lc n\:\}}
\def\Cxs{\C^\times}

\nc{\Il}{{\Ic_{\bla}}}
\nc{\Ikn}{{\Ic_{\:\kn}}}
\nc{\Ikon}{{\Ic_{\:k-1,\>n}}}
\nc{\Fla}{\Fc_\bla}
\nc{\tfl}{{T^*\<\Fla}}
\nc{\GL}{{G\<L_n(\C)}}
\nc{\GLC}{{\GL\times\C^*}}

\def\Gkn{\Gr(k,n)}
\def\TGkn{{\dsize T^*\Gkn}}
\def\HT{{\dsize H^*_T\bigl(\TGkn\<\bigr)}}
\def\KT{K_T\bigl(\TGkn\<\bigr)}

\def\CZH{\C[\:\ZZ^{\pm1}; H^{\pm1}]}

\def\qti{\tilde q}
\def\kat{\tilde\ka}

\def\Ibt{\tilde{\smash\Ibb\vp1}\vp\Ibb}

\def\acuv#1{\acute{#1}\vp{#1}}
\def\Edd{\skew4\acuv E}
\def\Fdd{\acuv F}

\def\Udd{\skew3\acuv{\smash U\vp i}\vp U}

\def\Xdd{\acuv X}
\def\Med{\skew2\acuv{\smash\Me\vp i}\vp\Me}
\def\Ted{\skew2\acuv{\smash\Te\vp i}\vp\Te}

\def\pii{\pi\sqrt{\<-1}}
\def\piit{\pi\:\sqrt{\<-1}}

\def\xxx{x_1\lc x_n}

\def\zzz{z_1\lc z_n}

\def\Dt{\skew3\Tilde{\vp1\smash D}}

\def\Pt{\skew5\Tilde{\vp1\smash P}}

\def\Dh{\Hat D}

\def\Phr{\Phi^{\raise.2ex\hbox{$\sssize\<\mathrm{r\<ed}\<$}}}

\def\les#1{\sssize<\<#1}
\def\ges#1{\sssize>\<#1}

\def\zero#1{\raise#1ex\hbox{$\sssize0$}}
\def\infs#1{\raise#1ex\hbox{$\sssize\infty$}}
\def\Co{C^{\>\zero{.15}}}
\def\Cf{C^{\:\infs{.26}}}
\def\Do{D^{\:\zero{.15}}}

\def\Fo{F^{\>\zero{.15}}}
\def\Ff{F^{\:\infs{.26}}}
\def\Go{G^{\:\zero{.20}}}
\def\Gf{G^{\infs{.32}}}
\def\cho{\chi^{\zero{-.31}}}
\def\chf{\chi^{\infs{-.20}}}
\def\ioo{\iota^{\<\zero{-.31}}}
\def\iof{\iota^{\?\infs{-.20}}}
\def\Mo{M^{\:\zero{.15}}}
\def\Mf{M^{\infs{.26}}}
\def\muo{\mu^{\:\zero{-.31}}}
\def\muf{\mu^{\infs{-.20}}}
\def\pio{\pi^{\:\zero{-.37}}}
\def\pif{\pi^{\infs{-.26}}}

\def\Pso{\Psi^{\:\zero{.15}}}
\def\Psf{\Psi^{\infs{.26}}}

\def\So{S^{\:\zero{.20}}}
\def\Sf{S^{\infs{.32}}}

\def\Weo{\We^{\:\zero{.15}}}
\def\Wef{\We^{\:\infs{.26}}}
\def\Weod{\skew3\acuv\We^{\:\zero{.15}}}

\def\Yo{Y^{\:\zero{.15}}}
\def\Yf{Y^{\:\infs{.26}}}

\def\Lp{L^{\!\raise.25ex\hbox{$\sssize+$}}}
\def\Lin{L^{\!\raise.2ex\hbox{$\sssize\mathrm{i\<n\<t}$}}}
\def\Wtr{W^{\:\raise.2ex\hbox{$\sssize\mathrm{t\<r}$}}}
\def\Wdd{\acuv W^{\:\raise.2ex\hbox{$\sssize\mathrm{t\<r}$}}}
\def\Psin{\Psi^{\raise.2ex\hbox{$\sssize\mathrm{i\<n\<t}\<$}}}
\def\sii{\si^{{\sssize\uparrow}\lower.2ex\hbox{$\ssize i$}}}

\def\sip{\si^{\:\lower.2ex\hbox{$\ssize\prime$}}}

\def\BBg{\BB_{\raise.16ex\hbox{$\sssize\<>\?1$}}}
\def\BBl{\BB_{\raise.16ex\hbox{$\sssize\?<\<1$}}}

\def\bul{\mathbin{\raise.2ex\hbox{$\sssize\bullet$}}}
\def\intt{\mathchoice
{\mathop{\raise.2ex\rlap{$\,\,\ssize\backslash$}{\intop}}\nolimits}
{\mathop{\raise.3ex\rlap{$\,\sssize\backslash$}{\intop}}\nolimits}
{\mathop{\raise.1ex\rlap{$\sssize\>\backslash$}{\intop}}\nolimits}
{\mathop{\rlap{$\sssize\:\backslash$}{\intop}}\nolimits}}

\def\GZ/{Gelfand-Zetlin}
\def\KZ/{{\slshape KZ\/}}
\def\qKZ/{{\slshape qKZ\/}}
\def\qKZB/{{\slshape qKZB\/}}
\def\XXX/{{\slshape XXX\/}}
\def\XXZ/{{\slshape XXZ\/}}

\def\hb{\bs h}
\def\lb{\bs l}

\def\bss{\bs s}

\nc{\bla}{{\bs\la}}

\def\zz{{\bs z}}

\def\TT{{\bs t}}

\def\Sym{\on{Sym}}

\def\Tb{\bs T}

\def\GG{{\bs\Gm}}

\def\xx{{\bs x}}

\def\Czon{\C\<\setminus\?\{0\:,\<1\}}
\def\CRp{\C\<\setminus\<\R_{\ge0}}

\def\naqla{\nabla^{\>\vp|\mathrm{\sssize quant}}}
\def\naknp{\naqla_{\?\kn,\>p}}

\def\TT{{\bs t}}

\def\YY{{\bs Y}}
\def\ZZ{{\bs Z}}
\def\zzst{{\tsize\zz^{\lower.02ex\hbox{$\<\ssize\star$}}}}
\def\zzss{\zz^{\raise.08ex\hbox{$\sssize\?\star$}}}
\def\zzs{\mathchoice{\rlap{$\zzst$}}{\rlap{$\zzst$}}{\rlap{$\zzss$}}
{\rlap{$\sssize\zz^{\<\star}$}}\phan{\zz}}

\begin{document}

\hrule width0pt
\vsk->

\title[Monodromy of the equivariant quantum differential equation]
{Monodromy of the equivariant quantum differential\\[2pt]
equation of the cotangent bundle of\\[2pt]
a Grassmannian}

\author[Vitaly Tarasov and Alexander Varchenko]
{Vitaly Tarasov$\>^\circ$ and Alexander Varchenko$\>^\star$}

\maketitle

\begin{center}
{\it $^{\star}\<$Department of Mathematics, University
of North Carolina at Chapel Hill\\ Chapel Hill, NC 27599-3250, USA\/}
\vsk.5>
{\it $\kern-.4em^\circ\?$Department of Mathematical Sciences,
Indiana University\,--\>Purdue University Indianapolis\kern-.4em\\
402 North Blackford St, Indianapolis, IN 46202-3216, USA\/}
\end{center}

{\let\thefootnote\relax
\footnotetext{\vsk-.8>\noindent
$^\circ\<${\sl E\>-mail}:\enspace vtarasov@iupui.edu\>,
supported in part by Simons Foundation grants \rlap{430235, 852996}
\\
$^\star\<${\sl E\>-mail}:\enspace anv@email.unc.edu\>,
supported in part by NSF grants DMS-1665239, DMS-1954266}}

\vsk>
{\leftskip3pc \rightskip\leftskip \parindent0pt \Small
{\it Key words\/}: Grassmannian, quantum differential equation,
$\:q\:$-hypergeometric solution
\vsk.6>
{\it 2010 Mathematics Subject Classification\/}: 14N35, 53D45, 14D05, 33C70
\par}

\begin{abstract}
We describe the monodromy of the equivariant quantum differential equation
of the cotangent bundle of a Grassmannian in terms of the equivariant
$\:K\?$-theory algebra of the cotangent bundle. This description is based
on the hypergeometric integral representations for solutions of the equivariant
quantum differential equation. We identify the space of solutions with
the space of the equivariant $\:K\?$-theory algebra of the cotangent bundle.
In particular, we show that for any element of the monodromy group, all entries
of its matrix in the standard basis of the equivariant $\:K\?$-theory algebra
of the cotangent bundle are Laurent polynomials with integer coefficients in
the exponentiated equivariant parameters.
\end{abstract}

\vsk1.2>
\rightline{\it In memory of Igor Krichever (1950\:--2022)}
\vsk-.9>
\strut

{\small\tableofcontents\par}

\setcounter{footnote}{0}
\renewcommand{\thefootnote}{\arabic{footnote}}

\vsk-2>\vsk>
\section{Introduction}
This paper has grown in 2019 from our attempts to understand what is written
in the paper \cite{AO} by M.\,Aganagic and A.\,Okounkov on the monodromy
of the quantum difference equations in the equivariant $\:K\?$-theory of
Nakajima varieties.\?{\def\thefootnote{$\diamond$}\footnotemark
\def\thefootnote{\kern-\parindent}\footnotetext{\rule{0pt}{.9\baselineskip}%
$^\diamond$\:When the second author of this paper was a student at
Kolmogorov's boarding high school in Moscow for mathematically gifted students,
his one-year-older friend Lenya Levin, see \cite{Le}, was teaching him
the basics of automata theory. Among other things Lenya was preaching that
{\it the language of a higher level machine sounds for a lower level machine
like white noise!\/}}}
We consider the simplest example of a Nakajima variety --- the cotangent bundle
$\,\TGkn\,$ of a Grassmannian $\,\Gkn\,$. In this example we describe the
monodromy of solutions of the quantum differential equation in the equivariant
cohomology of $\,\TGkn\,$, which is arguably a more difficult problem than
the description of the monodromy of the quantum difference equation in the
equivariant $\:K\?$-theory of $\,\TGkn\,$, since solutions of the differential
equations are multivalued functions while solutions of the difference equations
are univalued.

\vsk.2>
While the study of the monodromy of the difference equations in \cite{AO}
is based on the algebraic geometry of the moduli spaces of quasi-maps,
our tools are based on the identification, presented in \cite{GRTV,MO,RTV1},
of the equivariant quantum differential equations with the dynamical
differential equations in the theory of quantum integrable systems, on the one
hand, and based on the integral representations for solutions of the dynamical
equation constructed in \cite{TV1,MV,TV4}, on the other hand.

\vsk.2>
We consider the cotangent bundle $\,\TGkn\,$ with the natural torus
$\,{T=(\Cxs\<)^n\?\times\Cxs}$ action. The equivariant quantum differential
equation for $\,\TGkn\,$ is a linear system of ordinary differential equations
in one complex variable $\,p\,$, depending on the equivariant parameters
$\,\zz=(z_1,\dots,z_n)\:,\>h\,$. The equivariant quantum differential equation
has singular points at $\,p=0\:,1,\infty\,$. To describe the monodromy,
it is enough to describe the following three linear operators:
the operator $\,\muo\:$ of the monodromy of solutions around $\,p=0\,$,
the operator $\,\muf\:$ of the monodromy of solutions around $\,p=\infty\,$,
and the operator of the analytic continuation of solutions from a neighbourhood
of the point $\,p=0\,$ to a neighbourhood of the point $\,p=\infty\,$.

\vsk.2>
We parametrize solutions of the quantum differential equation
\vv.03>
by elements of the equivariant $\:K\?$-theory algebra $\,\KT\,$
see Section \ref{secsolK}. For each $\,\:X\?\in\<\KT\,$, we associate
\vv.03>
two solutions of the quantum differential equation, $\,\:\Pso_{\?X}\:$ and
$\,\:\Psf_{\?X}\>$. The solution $\,\:\Pso_{\?X}\:$ is first defined in
\vv.03>
a neighborhood of $\,p=0\,$ and then analytically continued to the universal
cover of $\,\:\Czon\,$, while the solution $\,\:\Psf_{\?X}\>$ is first
defined in a neighborhood of $\,p=\infty\,$ and then analytically continued
to the universal cover of $\,\:\Czon\,$. In Theorem \ref{monoP}, we describe
the monodromy operator $\,\muo\:$ in terms of the solutions $\,\:\Pso_{\?X}\:$,
and the monodromy operator $\,\muf\:$ in terms of the solutions
$\,\:\Psf_{\?X}\>$. In Theorem \ref{trans}, we show that the solutions
$\,\:\Pso_{\?X}\:$ defined in a neighborhood of the point $\,p=0\,$
and solutions $\,\:\Psf_{\?X}\>$ defined in a neighborhood of the point
$\,p=\infty\,$ are related via the transition map $\,\:\tau:\KT\to\:\KT\,$
introduced in Section \ref{sectrans},
\vvn-.5>
\be
\Psf_{\?\tau X}(\zz\:;h\:;p)\,=\,\Pso_{\?X}(\zz\:;h\:;p)\,.\kern-2em
\ee
This describes the analytic continuation of solutions of the quantum
differential equation from a neighbourhood of the point $\,p=0\,$
to a neighbourhood of the point $\,p=\infty\,$.

\vsk.2>
Furthermore, Theorem \ref{monoint} shows that for any element of the monodromy
group, all entries of its matrix in the standard basis of $\,\KT\,$ are Laurent
polynomials with integer coefficients in the exponentiated equivariant
parameters.

\vsk.2>
To prove Theorem \ref{trans}, we write an integral formula for solutions of
the quantum differential equation and express the result via the solutions
$\,\:\Pso_{\?X}\:$ and $\,\:\Psf_{\?X}\>$ parametrized by elements of $\,\KT\,$,
see Theorem \ref{thmPsin}. Theorems \ref{monoP}, \ref{trans}, \ref{monoint},
and \ref{thmPsin} are the main results of the paper.

\vsk.2>
The paper is organized as follows. In Section \ref{sQde}, we introduce the
equivariant quantum differential equation for $\,\TGkn\,$ and the equivariant
$\:K\?$-theory algebra $\,\KT\,$. In Section \ref{sectrig}, we consider
the trigonometric weight functions. In Section \ref{sec Int Rep}, we describe
solutions of the quantum differential equation following \cite{TV5}\:,
parametrized the solutions by elements of $\,\KT\,$, and prove
Theorem \ref{monoP} on the monodromy operators $\,\muo\:$ and $\,\muf\:$.
In Section \ref{secint}, we consider an integral formula for solutions of the
quantum differential equation, and prove Theorems \ref{thmPsin} and \ref{trans}.
In Section \ref{sec:pfs}, we prove technical results used in the main part
of the paper. We collect the facts about the Schur polynomials in Appendix.

\vsk.2>
We thank Andrey Smirnov for useful discussions. The second author thanks
Max Planck Institute for Mathematics in Bonn for hospitality in May--June 2022.

\section{Equivariant quantum differential equation}
\label{sQde}

\subsection{Equivariant cohomology of $\,\TGkn\,$}
\label{sec:eq}
Let $\,k,n\,$ be nonnegative integers such that $\,k\le n\,$. Consider
the Grassmannian $\,\:\Gkn\,$ of $\>k\:$-dimensional subspaces in $\,\:\C^n$
and the cotangent bundle $\,\TGkn\,$. We embed $\,\:\Gkn\,$ in $\,\TGkn\,$
as the zero section and consider the points of $\,\:\Gkn\,$ as points of
$\,\TGkn\,$.

\vsk.2>
Denote by $\,\Ikn\:$ the set of $\>k\:$-\:element subsets of
\vv.06>
$\,\:\onen\,$. Let $\,e_1\lc e_n\,$ be the standard basis of $\,\C^n$.
For any $\>I\?\in\Ikn\,$, let $\,F_I\<\in\Gkn\?\subset\TGkn\,$ be the span
of $\,\{\:e_i\<\ |\ i\<\in\<I\>\}$.

\vsk.2>
Let $\,A\subset\GL\,$ \,be the torus of diagonal matrices and
$\,T=A\times\Cxs$. The group $\,A\,$ acts on $\,\:\C^n$ and hence on
$\,\TGkn\,$. Let the group $\,\:\Cxs$ act on $\,\TGkn\,$ by multiplication in
\vv.06>
each fiber. We denote by $\>-\:h\,$ its $\,\:\Cxs\<$-weight. The fixed points
of the action of $\,T\,$ on $\,\TGkn\,$ are the points $\,F_I\<\in\TGkn\,$,
$\,\:I\?\in\Ikn\,$.

\vsk.2>
Let $\,E\,$ and $\,\Eb\,$ be the vector bundles over $\,\:\Gkn\,$ with fibers
over a point $\,F\in\Gkn$ being $\,F\>$ and $\,\:\C^n\!/F\>$, respectively.
Denote by $\,\gmb=\{\:\gm_1\lc\gm_k\:\}\,$ and
$\,\gmbb=\{\:\gmbr_1\lc\gmbr_{n-k}\:\}\,$ the respective set of the Chern roots
of $\,E\,$ and $\,\Eb\,$. Denote by $\,\zz=\{\:\zzz\:\}\,$ the Chern roots
corresponding to the factors of the torus $\,A\,$.

\vsk.2>
Consider the equivariant cohomology algebra $\,H^*_T(\TGkn\:;\C)\,$.
Then
\vvn.4>
\begin{align}
\label{Hrel}
& \HT\,={}
\\[-2pt]
\notag
&\!\?{}=\,\C[\gmb]^{\>S_k}\!\ox\C[\gmbb]^{\:S_{n-k}}\!\ox\C[\zz\:;h\:]\>
\Big/\Bigl\bra\,\prod_{i=1}^k\,(u-\gm_i)\,\prod_{j=1}^{n-k}\,(u-\gmbr_j)\,=\,
\prod_{a=1}^n\,(u-z_a)\Bigr\ket\,,\kern-.6em
\\[-15pt]
\notag
\end{align}
where $\,u\,$ is an indeterminate. Notice that $\,c_1(E)=\gm_1\lsym+\gm_k\,$
and $\,c_1(\Eb)=\gmbr_1\lsym+\gmbr_{n-k}\>$ are the equivariant first Chern
\vv.04>
classes of the vector bundles $\,E\,$ and $\,\Eb\,$, respectively. Moreover,
\vvn-.4>
\be
c_1(E)+c_1(\Eb)\,=\,z_1\lsym+z_n\,.
\ee

\subsection{Equivariant quantum differential equation}
\label{secmult}
The quantum
\vv.07>
multiplication by divisors on $\,\HT\,$ is described in \cite{MO}.
\vv-.05>
The fundamental equivariant cohomology classes of divisors on $\TGkn$
\vv.07>
are linear combinations of $\,c_1(E)\,$ and $\,c_1(\Eb)\,$.
The quantum multiplication $\,c_1(E)\>*_p:\HT\to\HT\,$ by the divisor
$\,c_1(E)\,$ depends on a parameter $\,p\,$ and is given in \cite{MO}\:,
\vv.06>
see Theorem~10.2.1 and Proposition~11.2.2 therein.
The parameter $\,p\,$ in this paper equals
\vv.07>
$\,q^{\:\ell}$ in \cite[Proposition~11.2.2\:]{MO}\:. The quantum multiplication
by $\,c_1(E)+c_1(\Eb)\,$ does not depend on $\,p\,$ and coincides with
the ordinary multiplication by $\,c_1(E)+c_1(\Eb)=z_1\lsym+z_n\,$.
Furthermore, the quantum multiplication by any polynomial $\,g(\zz\:;h)\,$
coincides with the ordinary multiplication by $\,g(\zz\:;h)\,$.

\vsk.3>
The quantum connection $\,\:\naknp\>$ for $\,\HT\,$ is defined by the formula
\vvn.1>
\beq
\label{nabla}
\naknp\:=\,p\>\frac\der{\der\:p}\>-\>c_1(E)*_p{},\kern-2em
\vv-.1>
\eeq
see \cite[\:(1.15)\:]{MO}\:. The equation
\vvn-.4>
\beq
\label{qde}
p\>\frac{\der\<f}{\der\:p}\,=\,c_1(E)*_p\<f\kern-1.6em
\vv.1>
\eeq
for horizontal sections for the quantum connection
is called the {\it equivariant quantum differential equation}.
It can be viewed as an ordinary differential equation in $\,\:p\,$
depending on $\,\zz,h\,$ as parameters.

\begin{lem}
\label{lem sing}
Equation \eqref{qde} depends holomorphically on $\,\zz\:,h\,$.
For fixed $\,\zz,h\>$, equation \eqref{qde} is Fuchsian with singular points
at $\,p=0\,$, $\,p=1\,$, and $\,p=\infty\,$.
\end{lem}
\begin{proof}
The statement follows from \cite[Proposition~11.2.2\:]{MO}\:.
\end{proof}

Integral representations for solutions of the quantum differential equation
were constructed in \cite{TV3} as multidimensional hypergeometric integrals.
In \cite{TV4}\,, \cite{TV5}\,, solutions were presented in another way
in the form of multidimensional hypergeometric functions. We remind
the construction from \cite{TV4}\:, \cite{TV5} in Section \ref{sec Int Rep}.

\vsk.2>
In this paper we will describe the monodromy of equivariant quantum differential
equation \eqref{qde}\:. To this end, we will describe the analytic continuation
of its solutions around the singular points $\,\:p=0\,$ and $\,\:p=\infty\,$,
and from a neighbourhood of $\,p=0\,$ to a neighbourhood of $\,p=\infty\,$.
For the precise statements, see Theorems \ref{monoP}, \ref{trans}.

\vsk.2>
For convenience of writing, instead of differential equation \eqref{qde}\:,
we will consider the {\it modified quantum differential equation\/}
\vvn.3>
\beq
\label{mqde}
p\>\frac{\der\<f}{\der\:p}\,=\,c_1(E)*_p\<f\:-\:
\frac{hp}{1-p}\:\min\:(0,n-2k)\>f\kern-1.6em
\vv.2>
\eeq
If $\>f(p)\,$ solves equation \eqref{mqde}\:,
then $\,(1-p)^{h\<\min\:(0,\>n\:-2k)}f(p)\,$ solves equation \eqref{qde}\:,
and vice versa.

\subsection{Equivariant $\:K\?$-theory of $\,\TGkn\,$}
\label{sec:equivK}
We consider the equivariant $\:K\?$-theory algebra $\,\KT\,$.
Our general reference is \cite[Ch.5]{ChG}.

\vsk.2>
Introduce new variables $\,\:\GG=\{\:\Gm_{\?1}\lc\Gm_{\<k}\}\,$,
$\,\GGb=\{\:\Gmb_{\?1}\lc\Gmb_{\<n-k}\}\,$, $\ZZ=\{\:Z_1\lc Z_n\:\}$
and $H$. Let $\,\:\C[\:\GG^{\pm1}]\,$, $\,\:\C[\:\GGb^{\pm1}]\,$,
$\,\:\C[\:\ZZ^{\pm1};H^{\pm1}\:]\,$ be the algebras of Laurent polynomials
in the respective variables, and
$\,\:\C[\:\GG^{\pm1}]^{\>S_k}$, $\,\:\C[\:\GGb^{\pm1}]^{\:S_{n-k}}\:$
\vvn.4>
the algebras of symmetric Laurent polynomials. Then
\begin{align}
\label{Krel}
& \KT\,={}
\\[-2pt]
\notag
\kern-2em{}=\,{}&\C[\:\GG^{\pm1}]^{\>S_k}\!\ox\C[\:\GGb^{\pm1}]^{\:S_{n-k}}
\!\ox\CZH\>\Big/\Bigl\bra\,
\prod_{i=1}^k\,(u-\Gm_{\?i})\,\prod_{j=1}^{n-k}\,(u-\Gmb_{\!j})\,=\,
\prod_{a=1}^n\,(u-Z_a)\Bigr\ket\,.\kern-.66em
\end{align}
Here $\,u\,$ is an indeterminate, the variables $\,\:\Gm_{\?1}\lc\Gm_{\<k}\>$
and $\,\:\Gmb_{\?1}\lc\Gmb_{\<n-k}\>$ correspond to (virtual) line bundles
denoted by the same letters such that $\,E=\Gm_{\?1}\<\lsym\oplus\Gm_{\<k}\>$
and $\,\Eb=\Gmb_{\?1}\<\lsym\oplus\Gmb_{\<n-k}\,$,
and the variables $\,Z_1\lc Z_n\:,H\,$ correspond to the factors
\vv.06>
of the torus $\,T=(\Cxs)^n\!\times\Cxs$. The algebra $\,\KT\,$
is a free module over $\,K_T(\mathrm{pt})=\C[\:\ZZ^{\pm1};H^{\pm1}\:]\,$.

\vsk.2>
Recall the fixed points $\,F_I\<\in\TGkn\,$, $\,\:I\?\in\Ikn\,$, of the action
\vv.04>
of $\,T\,$ on $\,\TGkn\,$. The restriction of the class $\,X\<\in\KT\,$
\vv.04>
to the point $\,F_I\>$ will be denoted by $\,X|_{\:I}^{}\>$.
In terms of the variables $\,\:\GG,\GGb\:$, the restriction
to the point $\,F_I\>$ amounts to the substitution
\vvn.4>
\be
\{\:\Gm_{\?1}\lc\Gm_{\<k}\}\,\mapsto\,\{\:Z_i\ |\ \:i\<\in\<I\>\}\,,\qquad
\{\:\Gmb_{\?1}\lc\Gmb_{\<n-k}\}\,\mapsto\,\{\:Z_j\ |\ \>j\<\not\in\<I\>\}\,.
\kern-1em
\vv.4>
\ee
Consider the {\it equivariant localization} map
\vvn.4>
\begin{align*}
\Loc\::\:\KT\: &{}\to\,K_T\bigl(\TGkn^T\:\bigr)\>=\!
\tbigoplus_{I\in\Ikn\!}\!K_T(F_I)\kern-2em
\\
\Loc\::\:X\> &{}\mapsto\<\tbigoplus_{I\in\Ikn\!}\!X|_{\:I}^{}\,.\kern-2em
\\[-15pt]
\end{align*}
The equivariant localization theory asserts that the map $\,\:\Loc\,\:$ is
\vv.06>
an embedding of algebras, see for example \cite[\:Chapter 5\:]{ChG}\:,
\cite[\:Appendix\:]{RK}\:. Moreover, an element
$\,\bigoplus_{\>I\in\Ikn\!}\?U_I\,$ is in the image of $\,\:\Loc\,\:$ if and
\vv.06>
only if for any $\,I\?\in\Ikn$ and any $\,i\:,j\<\in\<\onen\,$, $\,\:i\ne j\,$,
the difference $\,U_I\<-U_{s_{\ij}(I)}\,$ is divisible by $\,Z_i\<-Z_j\,$,
where $\,s_{\ij}\<\in\<S_n\>$ is the transposition of $\,i\:,j\,$.

\subsection{Transition map}
\label{sectrans}

Consider additional variables $\,\Tb=\{\:T_1\lc T_k\}\,$.
\vv.04>
For a Laurent polynomial $\,P(\Tb;\ZZ;H)\,$ symmetric in $\,T_1\lc T_k\:$,
define
the class $\,P(\GG;\ZZ;H)\in\<\KT\,$ via the substitution
\vv.04>
$\,T_i\mapsto\Gm_{\?i}\>$, $\,i=1\lc k\,$.
Define also the class $\,P(H\GG;\ZZ;H)\,$ via the substitution
$\,T_i\mapsto\<H\Gm_{\?i}\>$, $\,i=1\lc k\,$.

\vsk.3>
Let $\,\Pc\>$ be the space of Laurent polynomials $\,P(\Tb;\ZZ;H)\,$ that
\vv.1>
are symmetric polynomials in $\,\:T_1^{-1}\?\lc T_k^{-1}\:$ of degree
at most $\,n-1\,$ in each of $\,\:T_1^{-1}\?\lc T_k^{-1}$, and such that
\vvn.4>
\beq
\label{PHZ}
P(Z_a\:,H\?Z_a\:,T_3\:\lc T_k\:;\ZZ;H)\,=\,0\,,\qquad a=1\lc n\,.\kern-2em
\eeq

Set
\vvn-.5>
\beq
\label{ET}
E(\Tb;H)\,=\,\prod_{i=1}^k\,\prod_{\satop{j=1}{j\ne i}}^k\,(1-HT_i/T_j)\,.
\kern-1.8em
\eeq

\begin{prop}
\label{lemPEK}
For any $\,P\<\in\Pc$, there are unique classes
$\,\So_{\<P}\:,\Sf_{\<P}\<\in\<\KT\:$ such that
\vvn-.1>
\be
P(\GG;\ZZ;H)\,=\,E(\GG;H)\>\So_{\<P}\,,\qquad
P(H\GG;\ZZ;H)\,=\,E(\GG;H)\>\Sf_{\<P}\,.\kern-2em
\vv.5>
\ee
Abusing notation, we will write
\vv.4>
\be
\So_{\<P}\>=\,P(\GG;\ZZ;H)/E(\GG;H)\,,\qquad
\Sf_{\<P}\:=\,P(H\GG;\ZZ;H)/E(\GG;H)\,.\kern-1.8em
\vv.2>
\ee
\end{prop}
\begin{proof}
For $\,I=\{i_1\?\lsym<i_k\}\<\in\Ikn\,$, denote
\vv.088>
$\,\ZZ_{\<I}=(Z_{i_1}\lc Z_{i_k})\,$. Since $\,P(\Tb;\ZZ;H)\,$
is symmetric in $\,T_1\lc T_k\,$ and has property \eqref{PHZ}\:,
\vv.13>
the values at $\,\Tb=\ZZ_{\<I}\,$ of the rational functions
$\,P(\Tb;\ZZ;H)/E(\Tb;H)\,$, $\,P(H\Tb;\ZZ;H)/E(\Tb;H)\,$,
\vvn.4>
\be
\Go_{\?I}\<(\ZZ;H)\>=\>\frac{P(\ZZ_{\<I}\:;\ZZ;H)}{E(\ZZ_{\<I}\:;H)}\;,\qquad
\Gf_{\?I}\?(\ZZ;H)\>=\>\frac{P(H\?\ZZ_{\<I}\:;\ZZ;H)}{E(\ZZ_{\<I}\:;H)}\;,
\kern-2em
\vv.4>
\ee
are Laurent polynomials in $\,\ZZ,H\,$ for all $\,I\?\in\Ikn\,$.
\vv.08>
Therefore, there exist unique classes $\,\So_{\<P}\:,\Sf_{\<P}\<\in\<\KT\,$
such that $\,\So_{\<P}|_{\:I}^{}\<=\Go_{\?I}\,$,
$\,\:\Sf_{\<P}\<|_{\:I}^{}\<=\Gf_{\?I}\,$ for all $\,I\?\in\Ikn\,$.
\vv.07>
That proves Proposition \ref{lemPEK}.
\end{proof}

Consider the maps
\vvn.1>
\begin{gather}
\label{pimap}
\pio:\:\Pc\:\to\>\KT\,,\qquad P\:\mapsto\:P(\GG;\ZZ;H)/E(\GG;H)\,,
\kern-2em
\\[6pt]
\notag
\pif:\:\Pc\:\to\>\KT\,,\qquad P\:\mapsto\:P(H\GG;\ZZ;H)/E(\GG;H)\,.
\kern-2em
\\[-15pt]
\notag
\end{gather}

\begin{prop}
\label{lempi}
The maps $\,\:\pio$ and $\;\pif\<$ are isomorphisms of $\;\CZH$-modules
\vvn.12>
$\; \Pc$ and $\,\:\KT\,$.
\end{prop}
\noindent
Proposition \ref{lempi} is proved in Section \ref{pipfs}.

\vsk.4>
Define the {\it transition map} $\,\:\tau:\KT\to\:\KT\,$ by the rule
\vvn.3>
\beq
\label{tau}
\tau\,=\,\pif(\pio\<)^{-1}\:.\kern-1em
\vv.2>
\eeq

\begin{cor}
\label{lemmu}
The map $\,\:\tau\:$ is an automorphism of the $\;\CZH$-module $\,\KT\,$.
\end{cor}

Recall that $\,\KT\,$ is a free $\,\C[\:\ZZ^{\pm1};H^{\pm1}\:]\:$-\:module.
\vv.06>
Then given a basis of $\,\KT\,$, the maps $\,\:\tau\,$ and $\,\:\tau^{-1}$
are given by matrices with entries in $\,\C[\:\ZZ^{\pm1};H^{\pm1}\:]\,$.
The determinant $\,\:\det\tau\,$ does not depend on a choice of a basis
of $\,\KT\,$.

\begin{prop}
\label{dettau}
We have $\;\det\tau\:=\:
H^{\tsize\<-\frac{n\:(n-1)}2\<\binom{\:n-1\:}{k-1}}_{\vp1}$.
\end{prop}
\begin{proof}
The statement follows from formulae \eqref{tauXI}\:, \eqref{WofERdet}
\vv.07>
in Section \ref{orthdet} by extension of scalars from
$\,\C[\:\ZZ^{\pm1};H^{\pm1}\:]\,$ to rational functions
$\,\C(\:\ZZ^{\pm1};H^{\pm1}\:)\,$.
\end{proof}

\section{Trigonometric weight functions}
\label{sectrig}
\subsection{Definition}
\label{deftrig}
For a function $\,f(\Tb)\,$, denote
\vvn.4>
\be
\Sym_{\>\Tb}\<f(\Tb)\,=\,\sum_{\si\in S_k}\>f(T_{\si(1)}\lc T_{\si(k)})\,.
\kern-2em
\vv-.5>
\ee
Set
\vvn-.4>
\beq
\label{U}
U(\Tb;\ZZ;H)\,=\,\prod_{i=1}^k\,\biggl(\;
\prod_{a=1}^{i-1}\,\:(1-H\?Z_a/T_i)\?\prod_{b=i+1}^k\!\<(1-Z_{\:b}/T_i)\?
\prod_{j=i+1}^k\<\frac{1-HT_j/T_i}{1-T_j/T_i}\,\biggr)\>.\kern-2.5em
\vv.3>
\eeq
Define the {\it reduced trigonometric weight function}
\vvn.5>
\beq
\label{We}
\We(\Tb;\ZZ;H)\,=\,\Sym_{\>\Tb} U(\Tb;\ZZ;H)\,,\kern-2em
\vv.3>
\eeq
the {\it trigonometric weight function}
\vvn-.2>
\beq
\label{Weo}
\Weo\<(\Tb;\ZZ;H)\,=\,\We(\Tb;\ZZ;H)\,
\prod_{i=1}^k\:\prod_{a=k+1}^n\<(1-Z_a\:/T_i)\,,\kern-2em
\vv-.1>
\eeq
and the {\it opposite trigonometric weight function}
\beq
\label{Wef}
\Wef\<(\Tb;\ZZ;H)\,=\,H^{\>k(k-1)/2}\,\:\We(\Tb;\ZZ;H)\,
\prod_{i=1}^k\:\prod_{a=k+1}^n\<(1-H\?Z_a\:/T_i)\,.
\kern-1.6em
\vv.2>
\eeq
By the standard reasoning, the functions $\,\:\We,\Weo\?,\Wef$ are
polynomials in $\,\:T_1^{-1}\?\lc T_k^{-1}\?,\ZZ,H\>$.

\begin{lem}
\label{symWtr}
The functions $\;\We(\Tb;\ZZ;H)\,$, $\>\Weo\<(\Tb;\ZZ;H)\,$,
\vv.1>
$\>\Wef\<(\Tb;\ZZ;H)\>$ are symmetric in $\,Z_1\lc Z_k\,$, and
in $\,Z_{k+1}\lc Z_n\,$.
\end{lem}
\begin{proof}
The symmetry in $\,Z_1\lc Z_k\,$ follows from the identity
\vvn.4>
\begin{align*}
& (1-Z_{\:b}/T_i)\>(1-H\?Z_a/T_j)\>(T_i-HT_j)-
(1-Z_{\:b}/T_j)\>(1-H\?Z_a/T_i)\>(T_j-HT_i)\,={}\kern-.5em
\\[5pt]
& \!{}=\,(1-Z_a/T_i)\>(1-H\?Z_{\:b}/T_j)\>(T_i-HT_j)-
(1-Z_a/T_j)\>(1-H\?Z_{\:b}/T_i)\>(T_j-HT_i)\,.\kern-.5em
\\[-12pt]
\notag
\end{align*}
The symmetry in $\,Z_{k+1}\lc Z_n\,$ is manifest.
\end{proof}

\vsk.2>
For a permutation $\,\si\<\in\<S_n\,$, denote
\vv.07>
$\,\ZZ_{\?\si}=(Z_{\si(1)}\lc Z_{\si(n)})\,$.
Recall the space $\,\Pc\:$ defined in Section \ref{sectrans}.

\begin{prop}
\label{WWP}
For any $\;\si\<\in\<S_n\,$, we have $\,\Weo\<(\Tb;\ZZ_{\?\si};H)\<\in\Pc$
and $\,\:\Wef\<(\Tb;\ZZ_{\?\si};H)\<\in\Pc\>$.
\end{prop}
\begin{proof}
The statement follows from formulae \eqref{U}\,--\,\eqref{Wef} by inspection.
\end{proof}

For $\,I\?\in\Ikn\,$, let $\,\si_I\<\in\<S_n\,$ be the permutation
of minimal length such that
\vvn.5>
\be
I\,=\,\{\:\si_I(1)\lc\si_I(k)\:\}\,.\kern-2em
\vv.5>
\ee
Recall $\,\ZZ_{\?I}=(Z_{\si_I(1)}\lc Z_{\si_I(k)})\,$.
Recall the function $\,E(\Tb;H)\,$ given by \eqref{ET}\:. Set
\vvn.2>
\beq
\label{RZ}
R(\<\ZZ)\,=\,\prod_{a=1}^k\:\prod_{b=k+1}^n\<(1-Z_{\:b}/Z_a)\,.\kern-2em
\vv.2>
\eeq

\begin{lem}
\label{lemWeIJ}
For any $\,\:I,J\<\in\Ikn\,$, we have
\vvn.6>
\begin{gather}
\label{WeIJ}
\Weo\<(\ZZ_{\?I}\:;\ZZ_{\?\si_{\<J}};H)\,=\,
\dl_{\IJ}\,\:E(\ZZ_{\?I}\:;H)\,R(\ZZ_{\?\si_I}\<)\,,\kern-1.6em
\\[7pt]
\notag
\Wef\<(H\?\ZZ_{\?I}\:;\ZZ_{\?\si_{\<J} };H)\,=\,
\dl_{\IJ}\,\:E(\ZZ_{\?I}\:;H)\,R(\ZZ_{\?\si_I}\<)\,.\kern-1.6em
\\[-12pt]
\notag
\end{gather}
\end{lem}
\begin{proof}
The statement follows from formulae \eqref{U}\,--\,\eqref{Wef} by inspection.
\end{proof}

For any $\,X\?\in\<\KT\,$, set
\vvn.2>
\begin{align}
\label{WofX}
\Weo_{\!X}\<(\Tb;\ZZ;H)\,&{}=\<\sum_{I\in\>\Ikn\!}
\frac{\Weo\<(\Tb;\ZZ_{\?\si_I};H)}{R(\ZZ_{\?\si_I}\<)}\;X|_I^{}(\ZZ;H)\,,
\kern-2em
\\[8pt]
\notag
\Wef_{\!X}\<(\Tb;\ZZ;H)\,&{}=\<\sum_{I\in\>\Ikn\!}
\frac{\Wef\<(\Tb;\ZZ_{\?\si_I};H)}{R(\ZZ_{\?\si_I}\<)}\;X|_I^{}(\ZZ;H)\,.
\kern-2em
\\[-12pt]
\notag
\end{align}
By the standard reasoning, the functions $\,\:\Weo_{\!X},\Wef_{\!X}\,$
\vv.1>
are Laurent polynomials in $\,\Tb,\ZZ,H\>$. Furthermore,
$\,\:\Weo_{\!X}\:,\Wef_{\!X}\<\in\Pc\,$ by Proposition \ref{WWP}.

\vsk.3>
Recall the maps $\,\:\pio\>$ and $\,\:\pif\:$ given by \eqref{pimap}\:.

\begin{prop}
\label{piWX}
We have $\,\:\pio\:\Weo_{\!X}=X\>$ and $\;\pif\:\Wef_{\!X}\<=X\>$
for any $\,X\?\in\<\KT\,$.
\end{prop}
\begin{proof}
By formulae \eqref{WofX}\:, \eqref{WeIJ}\:, we have
\vvn.4>
\be
\Weo_{\!X}\<(\ZZ_{\<I}\:;\ZZ;H)=X|_I^{}(\ZZ;H)\,,\qquad
\Wef_{\!X}\<(H\?\ZZ_{\<I}\:;\ZZ;H)=X|_I^{}(\ZZ;H)\kern-2em
\vv.3>
\ee
for all $\,I\?\in\Ikn\,$. This proves the statement.
\end{proof}

\subsection{Proof of Propositions \ref{lempi}}
\label{pipfs}
The maps $\,\:\pio\>$ and $\,\:\pif\:$ are surjective
by Proposition \ref{piWX}. The proofs of injectivity of $\,\:\pio\>$ and
$\,\:\pif\:$ are similar. We will prove below that the map $\,\:\pio\>$
is injective.

\vsk.2>
The proof of injectivity of $\,\:\pio\>$ is by induction on $\,k\,$.
\vv.04>
If necessary, we will indicate the dependence on $\,k\,$ explicitly writing
$\,\:\pio_{\<k}\:,\>\Pc_k\:$, etc. To simplify writing, we will not show
the dependence on $\,\ZZ,H\,$ explicitly.

\vsk.3>
Let $\,k=1\,$. The space $\,\Pc_1\>$ consists of Laurent polynomials
\vv.06>
$\,P(T_1)\,$ that are polynomials in $\,T_1^{-1}\<$ of degree at most
$\,n-1\,$. If $\,\:\pio_1\:P=0\,$, then $\,P(Z_a)=0\,\:$ for all
$\,a=1\lc n\,$. Since $\,P\>$ vanishes at $\,n\,$ points, $\,P=0\,$
identically.

\vsk.2>
Let $\,k=2\,$. The space $\,\Pc_2\,$ consists of Laurent polynomials
\vv.07>
$\,P(T_1,T_2)\,$ that are symmetric polynomials in $\,T_1^{-1}\?,T_2^{-1}\<$
\vv.07>
of degree at most $\,n-1\,$ in each of $\,T_1^{-1}\?,T_2^{-1}\<$. In addition,
$\,P(Z_a\:,H\?Z_a)=0\,\:$ for all $\,a=1\lc n\,$. If $\,\:\pio_2\:P=0\,$,
then $\,P(Z_a\:,Z_{\:b})=0\,\:$ for all $\,a,b=1\lc n\,$, $\,a\ne b\,$.
Expand $\,P\>$ as a polynomial in $\,T_2^{-1}\<$,
\vvn.2>
\be
P(T_1,T_2)\,=\,\sum_{i=0}^{n-1}\,\Pt_{\?i}\:(T_1)\,T_2^{-\:i}\:.
\kern-1.6em
\vv.1>
\ee
Since $\,P\<\in\Pc_2\,$, we have $\,\Pt_{\?i}\<\in\Pc_1\>$
for all $\,i=1\lc n$. Then for any $\,a=1\lc n\,$, the conditions
\vvn-.2>
\be
P(Z_a\:,H\?Z_a)=0\,,\qquad P(Z_a\:,Z_{\:b})=0\,,\quad b\ne a\,,\kern-1em
\vv.4>
\ee
yield $\,\Pt_{\?i}\:(Z_a)=0\,\:$. Hence, $\,\Pt_{\?i}=0\,$ identically
for all $\,i=1\lc n$, and $\,P=0\,$ identically.

\vsk.3>
In general, any $\,P\<\in\Pc_k\,$ is a symmetric polynomial
in $\,\:T_1^{-1}\?\lc T_k^{-1}\:$ of degree at most $\,n-1\,$
in each of $\,\:T_1^{-1}\?\lc T_k^{-1}$.
Furthermore, the vanishing conditions \eqref{PHZ} yield
\vvn.5>
\beq
\label{PZZH}
P(\ZZ_{\<I}\:,H\?Z_a)\,=\,0\,,\qquad I\?\in\Ikon\,,\quad a\<\in\<I\,,
\kern-1.3em
\vv.2>
\eeq
If $\,\:\pio_k\:P=0\,$, then
\vvn.2>
\beq
\label{PZZ}
P(\ZZ_{\<I}\:,Z_{\:b})\,=\,0\,,\qquad I\?\in\Ikon\,,\quad b\not\in\<I\,.
\kern-1.2em
\vv.5>
\eeq
Expand $\,P\>$ as a polynomial in $\,T_k^{-1}\<$,
\vvn-.2>
\be
P(\Tb)\,=\,\sum_{i=0}^{n-1}\,\Pt_{\?i}\:(T_1\lc T_{k-1})\,T_k^{-\:i}\:.
\kern-1.6em
\vv.1>
\ee
Since $\,P\<\in\Pc_k\,$, we have $\,\Pt_{\?i}\in\Pc_{k-1}\>$
\vv.07>
for all $\,i=1\lc n$. Then for any $\,I\?\in\Ikon\,$,
conditions \eqref{PZZH}\:, \eqref{PZZ} \:yield $\,\Pt_{\?i}(\ZZ_{\<I})=0\,$.
\vv.07>
Hence, $\,\:\pio_{k-1}\:\Pt_{\?i}=0\,$, and $\,\Pt_{\?i}=0\,$
identically for all $\,i=1\lc n$, by the induction assumption.
Therefore, $\,P=0\,$ identically.
\qed

\subsection{Determinants}
\label{orthdet}
Using definition \eqref{tau} of
\vvn.06>
the transition map, ${\tau=\pif(\pio\<)^{-1}\:}$,
formulae \eqref{WofX}\:, and Proposition \ref{piWX}, we have
\vvn.3>
\begin{gather}
\label{tauXI}
(\tau X)|_I^{}(\ZZ;H)\,=\<\sum_{J\in\>\Ikn\!}
\frac{\Weo\<(H\?\ZZ_{\<I}\:;\ZZ_{\?\si_{\<J}};H)}
{E(\ZZ_{\<I}\:;H)\>R(\ZZ_{\?\si_{\<J}}\<)}\;X|_J^{}(\ZZ;H)\,,\kern-2em
\\[8pt]
\notag
(\tau^{-1}\<X)|_I^{}(\ZZ;H)\,=\<\sum_{J\in\>\Ikn\!}
\frac{\Wef\<(\ZZ_{\<I}\:;\ZZ_{\?\si_{\<J}};H)}
{E(\ZZ_{\<I}\:;H)\>R(\ZZ_{\?\si_{\<J}}\<)}\;X|_J^{}(\ZZ;H)\,.\kern-2em
\\[-12pt]
\notag
\end{gather}

\begin{cor}
\label{corth}
For any $\,\:I,J\<\in\Ikn\,$, we have
\vvn.5>
\begin{gather}
\label{Worth}
\sum_{K\?\in\:\Ikn\!}\!\frac{\Wef\<(\ZZ_{\<I}\:;\ZZ_{\?\si_{\<K}};H)\,
\Weo\<(H\?\ZZ_{\<K}\:;\ZZ_{\<\si_{\<J}};H)}
{E(\ZZ_{\<K}\:;H)\>R(\ZZ_{\?\si_{\<K}}\?)}\;=\,
\dl_{\IJ}\,\:E(\ZZ_{\?I}\:;H)\,R(\ZZ_{\?\si_I}\<)\,,\kern-1.4em
\\[7pt]
\notag
\sum_{K\?\in\:\Ikn\!}\!\frac{\Weo\<(H\?\ZZ_{\<I}\:;\ZZ_{\?\si_{\<K}};H)\,
\Wef\<(\ZZ_{\<K}\:;\ZZ_{\?\si_{\<J}};H)}
{E(\ZZ_{\<K}\:;H)\>R(\ZZ_{\?\si_{\<K}}\?)}\;=\,
\dl_{\IJ}\,\:E(\ZZ_{\?I}\:;H)\,R(\ZZ_{\?\si_I}\<)\,.\kern-1.4em
\\[-17pt]
\notag
\end{gather}
\end{cor}
\begin{proof}
The statement follows from formulae \eqref{tauXI}
\end{proof}

\vsk.2>
Denote $\,\:\ZZ^{-1}\?=(Z_1^{-1}\?\lc Z_n^{-1})\,$. Set
\vvn.3>
\begin{gather}
\label{DDD}
D(\<\ZZ)\,=\,\prod_{a=2}^n\,\prod_{b\:=1}^{a-1}\,(1-Z_a/Z_{\:b})
^{\<\binom{\:n-2\:}{k-1}}_{\vp1}\:,\qquad
\Dt(\ZZ\:;H)\,=\,\prod_{a=1}^n\>\prod_{\satop{b\:=1}{b\ne a}}^n\,
(1-H\?Z_a/Z_{\:b})^{\?\binom{\:n-2\:}{k-2}}_{\vp1}\:,\kern-1em
\\[5pt]
\notag
\Dh(\ZZ\:;H)\,=\,D(\<\ZZ)\,D(\<\ZZ^{-1})\,\Dt(\ZZ\:;H)\,.\kern-.8em
\\[-15pt]
\notag
\end{gather}

\begin{prop}
\label{lemdet}
We have
\vvn-.3>
\begin{align}
\label{Wofdet}
\det\:\bigl(\:
\Weo\<(H\?\ZZ_{\<I}\:;\ZZ_{\?\si_{\<J}};H)\<\bigr)_{\<\IJ\:\in\>\Ikn}\?
&{}=\,H^{\tsize\<-\frac{n\:(n-1)}2\<\binom{\:n-1\:}{k-1}}_{\vp1}\>
\Dh(\ZZ\:;H)\,,\kern-2em
\\[3pt]
\notag
\det\:\bigl(\:
\Wef\<(\ZZ_{\<I}\:;\ZZ_{\?\si_{\<J}};H)\<\bigr)_{\<\IJ\:\in\>\Ikn}\?
&{}=\,H^{\tsize\:\frac{n\:(n-1)}2\<\binom{\:n-1\:}{k-1}}_{\vp1}\>
\Dh(\ZZ\:;H)\,.\kern-2em
\\[-12pt]
\notag
\end{align}
\end{prop}
\begin{proof}
Denote $\,\Do\<(\ZZ\:;H)=\det\:\bigl(\:
\Weo\<(H\?\ZZ_{\<I}\:;\ZZ_{\?\si_{\<J}};H)\<\bigr)\,$.
\vvn.04>
Consider extra variables $\>\YY\<=(Y_1\lc Y_n)\,$. The determinant
$\;\det\:\bigl(\:\Weo\<(H\?\ZZ_{\<I}\:;\YY_{\!\!\<\si_{\<J}};H)\<\bigr)\,$
\vv.04>
has $\,\binom{\:n-2\:}{k-1}\,$ pairs of equal columns if $\>Y_i\<=\<Y_j\,$ and
$\,\binom{\:n-2\>}{k-1}\,$ pairs of coinciding rows if $\>Z_i\<=Z_j\,$. Hence,
$\;\det\:\bigl(\:\Weo\<(H\?\ZZ_{\<I}\:;\YY_{\!\!\<\si_{\<J}};H)\<\bigr)\,$
\vv.04>
is divisible by $\,D(\YY)\>D(\<\ZZ^{-1})\,$ as a Laurent polynomial
in $\,\YY\!,\ZZ,H\,\:$. Therefore, $\,\Do\<(\ZZ\:;H)\,$ is divisible
by $\,D(\<\ZZ)\>D(\<\ZZ^{-1})\,$.

\vsk.2>
By Proposition \ref{WWP}, $\,\:\Weo\<(H\?\ZZ_{\<I}\:;\ZZ_{\?\si_{\<J}};H)\,$
is divisible by $\,E(\ZZ_{\?I}\:;H)\,$ for every $\,I\?\in\Ikn\,$.
\vvn.08>
Hence, $\,\Do\<(\ZZ\:;H)\,$, is divisible by
\be
\prod_{I\in\>\Ikn\!}\?E(\ZZ_{\?I}\:;H)\,=\,\Dt\<(\ZZ\:;H)\,,\kern-1.5em
\vv.2>
\ee
and thus by $\,\Dh(\ZZ\:;H)\,$, because $\,D(\<\ZZ)\>D(\<\ZZ^{-1})\,$
and $\,\Dt\<(\ZZ\:;H)\,$ are coprime. Since multiplying
both $\,\Do\<(\ZZ\:;H)\,$ and $\,\Dh\<(\ZZ\:;H)\,\:$
by $\,(Z_1\ldots Z_n)^{\:(n-1)\binom{\:n-1\:}{k-1}}\,$
\vv.08>
give polynomials in $\ZZ$ of the same homogeneous degree, we have
\vvn-.4>
\be
\Do\<(\ZZ\:;H)\>=\,\Co\<(H)\>\Dh\<(\ZZ\:;H)\,,\kern-2em
\vv.2>
\ee
where $\,\Co\<(H)\,$ is a Laurent polynomials in $\,H\,$.
In a similar way, we have
\vvn.4>
\be
\det\:\bigl(\:
\Wef\<(\ZZ_{\<I}\:;\ZZ_{\?\si_{\<J}};H)\<\bigr)\,=\,
\Cf\<(H)\>\Dh\<(\ZZ\:;H)\,,\kern-2em
\vv.3>
\ee
where $\,\Cf\<(H)\,$ is a Laurent polynomials in $\,H\,$.

\vsk.3>
Since by Corollary \ref{corth}, $\;\Co\<(H)\,\Cf\<(H)\,=\,1\,$,
\vv.07>
we have $\,\:\Co\<(H)=\>\Co\<(1)\>H^d\>$ for some integer $\,d\,$.
Moreover, $\,\:\Co\<(1)=1\,$ by Lemma \ref{lemWeIJ}.

\vsk.1>
Finally, by inspection of the equality
\vvn.13>
$\,\:\Do\<(\ZZ^{-1};H^{-1})=\:\Co\<(H^{-1})\>\Dh\<(\ZZ^{-1};H^{-1})\,$
we obtain that
$\;\Co(H^{-1})=\Co(H)\>H^{\:n(n-1)\binom{\:n-1\:}{k-1}}\>$. Therefore,
\vvn.2>
\be
d\,=\>-\>\frac{n\>(n-1)}2\>\binom{n-1\:}{k-1}\,.
\vv-.4>
\ee
Proposition \ref{lemdet} is proved.
\end{proof}

\begin{cor}
We have
\vvn-.2>
\begin{align}
\label{WofERdet}
\det\:\biggl(\:
\frac{\Weo\<(H\?\ZZ_{\<I}\:;\ZZ_{\<\si_{\<J}};H)}
{E(\ZZ_{\<I}\:;H)\>R(\ZZ_{\?\si_{\<J}}\<)}\:\biggr)_{\!\?\IJ\:\in\>\Ikn}\!\<
&{}=\,H^{\tsize\<-\frac{n\:(n-1)}2\<\binom{\:n-1\:}{k-1}}_{\vp1},\kern-2em
\\[6pt]
\notag
\det\:\biggl(
\frac{\Wef\<(\ZZ_{\<I}\:;\ZZ_{\?\si_{\<J}};H)}
{E(\ZZ_{\<I}\:;H)\>R(\ZZ_{\?\si_{\<J}}\<)}\biggr)_{\!\?\IJ\:\in\>\Ikn}\!\<
&{}=\,H^{\tsize\:\frac{n\:(n-1)}2\<\binom{\:n-1\:}{k-1}}_{\vp1}.\kern-2em
\\[-14pt]
\notag
\end{align}
\end{cor}
\begin{proof}
Straightforward from formulae \eqref{ET}\:, \eqref{RZ}\:, \eqref{DDD}\:,
\eqref{Wofdet}\:.
\end{proof}

\section{Solutions of quantum differential equation}
\label{sec Int Rep}
In this section, we remind the construction of solutions of quantum
differential equation \eqref{qde} given in \cite{TV3}\:, \cite{TV4}\:.
In the notation of those papers, we consider the case $\,N=2\,$,
$\bla=(k,n-k)\,$, $\,q_1\<=p^{-1}$, $\,q_2=1\,$, $\,\ka=\:-\:1\,$,
$\,\qti_1\<=p$,
\vv.07>
$\,\qti_2=1\,$, $\,\kat=1\,$. We adapt formulae from \cite{TV3}\:, \cite{TV4}
\vv.06>
to the notation used in this paper. We also eliminate or modify certain factors
that are irrelevant for our consideration.

\subsection{Weight function and master function}
\label{prelim}
Introduce new variables $\,\:\TT=\{t_1,\dots,t_k\}\,$.
For a function $\,f(\TT)\,$, denote
\vvn-.3>
\be
\Sym_{\>\TT}\<f(\TT)\,=\,\sum_{\si\in S_k}f(t_{\si(1)}\lc t_{\si(k)})\,.
\kern-2em
\vv-.5>
\ee
Define the {\it reduced weight function\/}
\vvn.3>
\begin{align}
\label{W}
& W\<(\TT\:;\gmb\:;\gmbb\:;h)\,={}
\\[4pt]
\notag
&\!{}=\,\Sym_{\>\TT}\biggl(\;\prod_{i=1}^k\,\biggl(\;
\prod_{l=1}^{i-1}\;(t_i\<-\gm_l\<-h)\?\prod_{m=i+1}^{\:k}\?(t_i\<-\gm_m)
\<\prod_{j=i+1}^k\<\frac{t_i\<-t_j\<-h}{t_i-t_j}\,\biggr)\?\biggr)\>
\prod_{s=1}^{n-k}\,(t_i\<-\gmbr_{\?s})\,.\kern-1em
\\[-15pt]
\notag
\end{align}
Clearly, $\,W\<(\TT\:;\gmb\:;\gmbb\:;h)\,$
\vv.2>
is a polynomial in $\,\:\TT\,$.

\begin{lem}
\label{symW}
The function $\,W\<(\TT\:;\gmb\:;\gmbb\:;h)$ is symmetric
in $\gm_1\lc \gm_k\,$.
\end{lem}
\begin{proof}
The statement follows from the identity
\vvn.2>
\begin{align*}
& (t_i\<-\gm_m)\>(t_j\<-\gm_l\<-h)\>(t_i\<-t_j\<-h)-
(t_j\<-\gm_m)\>(t_i\<-\gm_l\<-h)\>(t_j\<-t_i\<-h)\,={}
\\[4pt]
&\!{}=\,(t_i\<-\gm_l)\>(t_j\<-\gm_m\<-h)\>(t_i\<-t_j\<-h)-
(t_j\<-\gm_l)\>(t_i\<-\gm_m\<-h)\>(t_j\<-t_i\<-h)\,.\kern-1em
\\[-34pt]
\end{align*}
\end{proof}
\vsk-.4>
Set
\vvn-.7>
\beq
\label{QG}
Q(\gmb\:;\gmbb\:;h)\,=\,
\prod_{i=1}^k\,\prod_{j=1}^{n-k}\,(\gm_i\<-\gmbr_{\?j}\<-h)\,.\kern-1em
\eeq
The product $\,W\<(\TT\:;\gmb\:;\gmbb\:;h)\>Q(\gmb\:;\gmbb\:;h)\,$
is called the {\it cohomological weight function}.

\vsk.4>
Define the {\it reduced master function}
\vvn.4>
\begin{align}
\label{Phr}
\Phr\< &{}(\TT\:;\zz\:;h)\,=\,e^{\:n\:\pii\>\sum_{j=1}^{\:k}\?t_j}\<\times{}
\\[4pt]
\notag
&{}\!\!\times\>\prod_{i=1}^k\,\biggl(\,\prod_{\satop{j=1}{j\ne i}}^k
\,\frac1{\Gm(t_j\<-t_i)\,\Gm(1+t_i\<-t_j\<-h)}\,\:
\prod_{a=1}^n\,\Gm(z_a\<-t_i)\,\Gm(t_i\<-z_a\<-h)\<\biggr)\,,\kern-1.15em
\\[-16pt]
\notag
\end{align}
where $\,\:\Gm\>$ is the Gamma\:-\:function, and the {\it master function}
\vvn.4>
\beq
\Phi(\TT\:;\zz\:;h\:;p)\,=\,
p^{\>\sum_{j=1}^{\:k}\?t_j}\>\Phr\<(\TT\:;\zz\:;h)\,.\kern-1.8em
\vv.3>
\eeq

\begin{example}
Let $\,k=1\,$. Then
\begin{gather*}
W\<(\TT\:;\gmb\:;\gmbb\:;h)\>=\>\prod_{i=1}^{n-k}\,(t_1\<-\gmbr_i)\,,
\qquad
Q(\gmb\:;\gmbb\:;h)\>=\>\prod_{i=1}^{n-k}\,(\gm_1\<-\gmbr_i\<-h)\,,
\kern-1.4em
\\[1pt]
\Phr\<(\TT\:;\zz\:;h)\>=\>e^{\:n\:\pii\,t_1}\:
\prod_{a=1}^n\,\Gm(z_a\<-t_1)\,\Gm(t_1\<-z_a\<-h)\,,\kern1.6em
\Phi(\TT\:;\zz\:;h\:;p)\>=\>p^{\>t_1}\:\Phr\<(\TT\:;\zz\:;h)\,.
\end{gather*}
\end{example}

\subsection{Solutions $\,\:\Pso_{\?I}(\zz\:;h\:;p)\,$
and $\;\Psf_{\?I}\<(\zz\:;h\:;p)\,$}
\label{secsolI}
For a function $\,f(\TT)\,$ and
\vv.07>
a point $\bss=(s_1\lc s_k)\,$,
\vvn.2>
let $\,\:\Res_{\>\TT\:=\:\bss}f(\TT)\,$ be the iterated residue,
\beq
\label{Res}
\Res_{\>\TT\:=\:\bss}f(\TT)\,=\,
\Res_{\>t_1=\:s_1}\ldots\>\Res_{\>t_k=\:s_k}f(\TT)\,.\kern-2em
\vv.4>
\eeq
For a subset $\,I=\{i_1\?\lsym<i_k\}\<\in\Ikn\,$, denote
\vvn.4>
\be
\zz_I\>=\,(z_{i_1}\lc z_{i_k})\,,\qquad
\zz_I\<+\hb\,=\,(z_{i_1}\?+h\lc z_{i_k}\?+h)\,.\kern-2em
\vv.3>
\ee
Consider the series
\vvn.4>
\begin{gather}
\label{Fo}
\Fo_{\?I}(\zz\:;h\:;p)\,=\,\bigl(\<-\:\Gm(\<-\:h)\<\bigr)^{-k}
\sum_{\lb\in\:\Z_{\ge 0}^k\!}\Res_{\>\TT\:=\zz_I\<+\:\lb}
\bigl(\Phr\<(\TT\:;\zz\:;h)\>W\<(\TT\:;\gmb\:;\gmbb\:;h)\<\bigr)\,
p^{\>\sum_{i=1}^{\:k}\<l_i}\:,\kern-2em
\\[4pt]
\label{Ff}
\Ff_{\?I}(\zz\:;h\:;p)\,=\:\bigl(\:\Gm(\<-\:h)\<\bigr)^{-k}
\sum_{\lb\in\:\Z_{\ge 0}^k\!}\Res_{\>\TT\:=\zz_I\<+\hb-\:\lb}
\bigl(\Phr\<(\TT\:;\zz\:;h)\>W\<(\TT\:;\gmb\:;\gmbb\:;h)\<\bigr)\,
p^{-\!\sum_{i=1}^{\:k}\<l_i}\:.\kern-2em
\end{gather}

\vsk.2>
Let $\,L\,$ be the complement in $\,\:\C^n\!\times\C\,$ of the hyperplanes
\vvn.4>
\beq
\label{L}
z_a\<-z_b\<\in\Z\,,\qquad z_a\<-z_b\<-h\in\Z_{\le 0}\,,\qquad a\ne b\,.
\kern-2em
\vv.3>
\eeq

\begin{prop}
\label{p<>1}
For any $\,I\?\in\Ikn\,$, the function $\,\Fo_{\?I}(\zz\:;h\:;p)\:$ is
\vv.07>
holomorphic in $\,\:p\:$ for $\,|\:p\:|\<<1\,$, and the function
$\,\Ff_{\?I}(\zz\:;h\:;p)\:$ is holomorphic in $\,\:p\:$ for $\,|\:p\:|>1\,$.
\vv.07>
The functions $\,\Fo_{\?I}(\zz\:;h\:;p)\:$ and $\,\Ff_{\?I}(\zz\:;h\:;p)\:$ are
holomorphic in $\,\zz,h\>$ for $\,(\zz,h)\in L\,$.
\end{prop}
\noindent
The statement is proved in Section \ref{sec:props}.

\vsk.4>
Given $\,I\?\in\Ikn\,$, set
\begin{gather}
\label{Psof}
\Pso_{\?I}(\zz\:;h\:;p)\,=\,p^{\>\sum_{i\in I}\<z_i}\>
\Fo_{\?I}(\zz\:;h\:;p)\,Q(\gmb\:;\gmbb\:;h)\,,\kern-1.8em
\\[6pt]
\notag
\Psf_{\?I}\<(\zz\:;h\:;p)\,=\,p^{\>kh\:+\sum_{i\in I}\<z_i}\>
\Ff_{\?I}(\zz\:;h\:;p)\,Q(\gmb\:;\gmbb\:;h)\,.\kern-1.8em
\\[-13pt]
\notag
\end{gather}
By Proposition \ref{p<>1}, the function $\,\:\Pso_{\?I}(\zz\:;h\:;p)\,$
is holomorphic on the universal cover of the punctured unit disk
$\,\:0<|\:p\:|\<<1\,$, and the function $\,\:\Psf_{\?I}(\zz\:;h\:;p)\,$
is holomorphic on the universal cover of the domain $\,\:|\:p\:|>1\,$.

\begin{thm}
\label{PsIsol}
For any $\,I\?\in\Ikn\,$, the functions $\,\:\Pso_{\?I}(\zz\:;h\:;p)\:$
and $\;\Psf_{\?I}\<(\zz\:;h\:;p)\:$ are multivalued solutions of differential
equation \eqref{mqde} for $\,\:0<|\:p\:|\<<1\:$ and $\;|\:p\:|>1\,$,
respectively.
\end{thm}
\begin{proof}
For the function $\,\:\Pso_{\?I}(\zz\:;h\:;p)\,$, the statement follows from
\cite[Theorem 9.5\:]{TV4}\:. The proof of the statement for the function
$\,\:\Psf_{\?I}(\zz\:;h\:;p)\,$ is similar.
\end{proof}

\begin{cor}
\label{corPsI}
For any $\,(\zz,h)\in L\:$ and $\,I\?\in\Ikn\,$, the functions
\vv.06>
$\,\:\Pso_{\?I}(\zz\:;h\:;p)\:$ and $\;\Psf_{\?I}\<(\zz\:;h\:;p)\:$ can be
analytically continued as functions of $\,\:p\:$ to holomorphic functions
on the universal cover of $\,\:\Czon\,$.
\end{cor}
\begin{proof}
The statement follows from Theorem \ref{PsIsol} and Lemma \ref{lem sing}.
\end{proof}

\subsection{Solutions and $\:K\?$-theory}
\label{secsolK}
For $\,\:X\?\in\<\KT\,$, recall its restrictions
\vvn.16>
$\,\:X|_I(\ZZ;H)\,$, $\,I\?\in\Ikn\,$, to the fixed points. Set
\vvn.2>
\be
\Xdd|_I(\zz\:;h)\,=\,X|_I\bigl(e^{\:2\pii\,z_1}\?\lc e^{\:2\pii\,z_n}\:;
e^{\:2\pii\,h}\:\bigr)\,.\kern-2em
\vv.33>
\ee
Define the functions $\,\:\Pso_{\?X}(\zz\:;h\:;p)\,$
and $\,\Psf_{\?X}\<(\zz\:;h\:;p)\,$ by the rule
\vvn.5>
\begin{gather}
\label{PsofX}
\Pso_{\?X}(\zz\:;h\:;p)\,=\!\sum_{I\in\>\Ikn\!}\?
\Xdd|_I(\zz\:;h)\>\Pso_{\?I}(\zz\:;h\:;p)\,,\kern-1.5em
\\[4pt]
\notag
\Psf_{\?X}\<(\zz\:;h\:;p)\,=\!\sum_{I\in\>\Ikn\!}\?
\Xdd|_I(\zz\:;h)\>\Psf_{\?I}(\zz\:;h\:;p)\,.\kern-1.5em
\\[-15pt]
\notag
\end{gather}
By Corollary \ref{corPsI}, for any $\,(\zz,h)\in L\,$, both
\vv.07>
$\,\:\Pso_{\?X}(\zz\:;h\:;p)\:$ and $\,\:\Psf_{\?X}\<(\zz\:;h\:;p)\:$
as functions of $\,\:p\:$ are holomorphic functions on the universal cover of
$\,\:\C\<\setminus\?\{0,1\}\,$.

\vsk.3>
Let $\,\Lp\:$ be the complement in $\,\:\C^n\!\times\C\,$ of the hyperplanes
\vv.3>
\beq
\label{Lp}
z_a\<-z_b\<-h\in\Z_{\le 0}\,,\qquad a\ne b\,.\kern-2em
\vv.4>
\eeq

\begin{thm}
\label{regz=z}
For any $\,(\zz,h)\in\Lp\!$ and $\,\:X\?\in\<\KT\,$, the functions
\vv.07>
$\,\:\Pso_{\?X}(\zz\:;h\:;p)\:$ and $\;\Psf_{\?X}\<(\zz\:;h\:;p)\:$
\vv.07>
are solutions of differential equation \eqref{mqde}\:. Moreover,
$\,\:\Pso_{\?X}(\zz\:;h\:;p)\:$ and $\,\:\Psf_{\?X}\<(\zz\:;h\:;p)\:$
are holomorphic functions of $\,\:\zz,h\:$ in $\,\Lp\<$.
\end{thm}
\begin{proof}
By Proposition \ref{p<>1} and Theorem \ref{PsIsol}, the functions
\vv.05>
$\,\:\Pso_{\?X}(\zz\:;h\:;p)\,$ and $\;\Psf_{\?X}\<(\zz\:;h\:;p)\,$ are
solutions of differential equation \eqref{mqde} holomorphic in $\,L\,$.
\vv.07>
The fact that $\,\:\Pso_{\?X}(\zz\:;h\:;p)\,$ is regular at the hyperplanes
\vv.07>
$\,z_a\<-z_b\in\Z\,\:$ for all $\,a\ne b\,$, and hence, extends to $\,\Lp\<$,
follows from \cite[Proposition~4.13\:]{TV5} since the first of formulae
\vv.04>
\eqref{PsofX} matches formula \cite[(4.41)\:]{TV5}\:. The proof of
the regularity of $\;\Psf_{\?X}\<(\zz\:;h\:;p)\,$ at the hyperplanes
\vv.07>
$\,z_a\<-z_b\in\Z\,\:$ for all $\,a\ne b\,$, is similar to the proof
of \cite[Proposition~4.13\:]{TV5}\:.
\end{proof}

Consider the polynomials $\,V_{\<I}\,$, $\,I\?\in\Ikn\,$,
defined by formula \eqref{Sch}\:.
\begin{prop}
\label{VK}
The classes $\,\:V_{\<I}(\GG)\,$, $\>I\?\in\Ikn\,$, \>give a basis
\vv.04>
of $\,\KT$ as a free module over $\,\C[\:\ZZ^{\pm1};H^{\pm1}\:]\,$.
\end{prop}
\begin{proof}
The statement follows from formula \eqref{Sch2} and Corollary \ref{lemfV}.
\end{proof}

\vsk.1>
Denote $\,\:\Hc_{\zz_0,\:h_0}\?=\HT\big/\bra\:\zz=\zz_0\,,\,h=h_0\:\ket\,$.

\begin{thm}
\label{KSiso}
For any $\,(\zz_0\:,h_0)\in\Lp$, the functions
\vv.04>
$\,\:\Pso_{\?V_{\?I}\<(\<\GG\<)}\?(\zz_0\:;h_0\:;p)\,$, $\,\:I\?\in\Ikn\,$,
form a basis of solutions of quantum differential equation \eqref{mqde}\:.
\vv.04>
Similarly, for any $\,(\zz_0\:,h_0)\in\Lp$, the functions
$\,\:\Psf_{\?V_{\?I}\<(\<\GG\<)}\?(\zz_0\:;h_0\:;p)\,$, $\,\:I\?\in\Ikn\,$,
\vv-.09>
form a basis of solutions of quantum differential equation \eqref{mqde}\:.
\end{thm}
\begin{proof}
The statement follows from Propositions \ref{VK}, \ref{lemdetM}, and
Lemma \ref{VH}.
\end{proof}

Informally, Theorem \ref{KSiso} means that each of the assignments
\vv.04>
$\,X\mapsto\Pso_{\?X}(\zz\:;h\:;p)\,$ and $\,X\mapsto\Psf_{\?X}(\zz\:;h\:;p)\,$
\vv.07>
parametrize solutions of quantum differential equation \eqref{mqde} by elements
of $\,\KT\,$.

\subsection{Monodromy of solutions}
\label{secmono}
Denote by $\,\:\Ub\,$ be the universal cover of $\,\:\Czon\,$. For a function
$\,f\>$ on $\,\:\Ub\,$, we define below functions $\,\:\muo\?f\>$ and
$\,\:\muf\?f\>$ on $\,\:\Ub\,$ called the monodromies of $\,f\>$ around zero
and infinity, respectively.

\vsk.2>
For a point $\,p\<\in\<\Ub\,$, denote by $\,p_*\:$ its projection to $\,\:\C\,$.
Let $\,\cho_p\,$ and $\,\chf_p\:$ be the paths on $\,\:\Ub\,$ starting at
$\,p\,$ that have the following properties. The projection of $\,\cho_p\,$ to
$\,\:\C\,$ is a simple loop that starts and ends at $\,p_*\>$, encircles zero
counterclockwise, and avoids the ray $\,\:[\:1,\infty)\<\subset\R\,$.
The projection of $\,\chf_p\:$ to $\,\:\C\,$ is a simple loop that starts
and ends at $\,p_*\>$, and encircles the segment $\,\:[\:0,1\:]\subset\R\,$
counterclockwise. Denote by $\,\ioo p\,$ and $\,\iof p\,$ the end points of
the paths $\,\cho_p\,$ and $\,\chf_p\:$, respectively. The points $\,\ioo p\,$
and $\,\iof p\,$ do not depend on the choice of the paths $\,\cho_p\:$ and
$\,\chf_p\:$.

\vsk.2>
For a function $\,f\>$ on $\,\:\Ub\,$, define the functions $\,\:\muo\?f\>$
and $\,\:\muf\?f\>$ by the rule
\vvn.2>
\be
\muo\?f(p)\>=\>f(\ioo p)\,,\qquad \muf\?f(p)\>=\>f(\iof p)\,.\kern-2em
\vv.3>
\ee
The maps $\,\:\muo\>$ and $\,\:\muf\:$ generate the monodromy group of
quantum differential equation \eqref{mqde}\:. We will describe the action
of $\,\:\muo\>$ and $\,\:\muf\:$ on solutions $\,\:\Pso_{\?X}\,$ and
$\,\:\Psf_{\?X}\>$ labeled by elements of $\,\KT\,$, see Theorems \ref{monoP},
\ref{trans}. According to Theorem \ref{KSiso}, this completely describes
the monodromy of differential equation \eqref{mqde}\:.

\vsk.2>
We begin with a specification of the definition of the solutions
$\,\Pso_{\?I}\>$ and $\;\Psf_{\?I}\<$, see \eqref{Psof}\:, as functions of
$\,p\<\in\<\Ub\,$. Fix a section of $\,\:\Ub\,$ continuous on $\,\:\CRp\,$.
Denote by $\,\:\BB\,$ the image of $\,\:\CRp\>$ under this section. For
\vv.04>
$\,p\in\BB\,$, we define $\,\:\arg\:p\,$ to be the value of $\,\:\arg\:p_*\:$
such that $\,\:-\:2\:\pi\<<\<\arg\:p<0\,$ and define fractional powers
$\,\dsize p^{\:s}\<=|\:p\:|^s e^{\:s\arg p\>\sqrt{\<-1}}$ accordingly. Set
\vvn.4>
\be
\BBl\:=\>\{\:p\in\BB\,,\;\;0<|\:p\:|\<<1\:\}\,,\qquad
\BBg\:=\>\{\:p\in\BB\,,\;\;|\:p\:|>1\:\}\,.\kern-2em
\vv.3>
\ee
Formulae \eqref{Psof}\:, \eqref{Fo}\:, \eqref{Ff} define the functions
\vv.04>
$\,\:\Pso_{\?I}\>$ for $\,p\in\BBl\:$ and the functions $\;\Psf_{\?I}\:$
for $\,p\in\BBg\>$. Then $\,\:\Pso_{\?I}\>$ and $\;\Psf_{\?I}\<$ are extended
as functions for $\,p\<\in\<\Ub\,$ by analytic continuation.
The solutions $\,\:\Pso_{\?X}\,$ and $\,\:\Psf_{\?X}\>$ for $\,X\<\in\KT\,$
are defined as functions of $\,p\<\in\<\Ub\,$ by formulae \eqref{PsofX}\:.

\begin{prop}
\label{monoI}
For any $\,I\?\in\Ikn\,$, we have
\vvn.4>
\be
\muo\Pso_{\?I}=\:\Pso_{\?I}\,\:\prod_{i\in I}\>e^{\:2\pii\,z_i}
\quad\text{and}\quad\,\muf\Psf_{\?I}\<=\:
\Psf_{\?I}\>e^{\:2\pii\,kh}\>\prod_{i\in I}\>e^{\:2\pii\,z_i}\:.\kern-2em
\ee
\end{prop}
\begin{proof}
Formulae \eqref{Psof} and Proposition \ref{p<>1} yield the first equality
on $\,\:\BBl\:$ and the second equality on $\,\:\BBg\>$. The statement holds
on the universal cover $\,\:\Ub\,$ by analytic continuation.
\end{proof}

\begin{thm}
\label{monoP}
For any $\,X\?\in\<\KT\,$, we have
\vvn.3>
\be
\muo\Pso_{\?X}=\:\Pso_{\Gm_{\!1}\<\ldots\:\Gm_{\!k}X}\quad\text{and}\quad\,
\muf\Psf_{\?X}=\:\Psf_{\?H^k\Gm_{\!1}\<\ldots\:\Gm_{\!k}X}\,.
\vv.2>
\ee
\end{thm}
\begin{proof}
The statement follows from formulae \eqref{PsofX} and Proposition \ref{monoI}.
\end{proof}

\vsk.5>
Recall the transition map $\,\:\tau:\KT\to\:\KT\,$, see \eqref{tau}\:.

\begin{thm}
\label{trans}
We have $\;\Psf_{\?\tau X}(\zz\:;h\:;p)=\>\Pso_{\?X}(\zz\:;h\:;p)\,$
for any $\,X\?\in\<\KT\,$.
\end{thm}
\noindent
Theorem \ref{trans} is proved in Section \ref{secpf}.

\vsk.3>
Consider the classes $\,\:V_{\<I}(\GG)\,$, $\>I\?\in\Ikn\,$. Define
\vvn-.1>
the matrices $\,\:\Me(\<\ZZ)=\bigl(\:\Me_{\IJ}(\<\ZZ)\bigr)_{\<\IJ\in\>\Ikn}\?$
and $\,\:\Te(\<\ZZ\:;H)=\bigl(\:\Te_{\!\IJ}(\<\ZZ\:;H)\bigr)_{\<\IJ\in\>\Ikn}$
as follows
\vvn.3>
\beq
\label{MeTe}
\Gm_{\?1}\?\ldots\Gm_{\<k}\,V_{\<I}(\GG)\,=\<
\sum_{J\in\>\Ikn\!}\?\Me_{\IJ}(\<\ZZ)\>V_{\!J}(\GG)\,,\qquad
\tau\bigl(V_{\<I}(\GG)\bigr)\,=\<
\sum_{J\in\>\Ikn\!}\?\Te_{\!\IJ}(\<\ZZ\:;H)\>V_{\!J}(\GG)\,.
\vv.1>
\eeq

\begin{prop}
\label{lemMeTe}
The matrix $\,\:\Me(\<\ZZ)$ is a Laurent polynomial in $\,\:\ZZ$ with
\vvn.03>
integer coefficients. The matrices $\,\:\Te(\<\ZZ\:;H)$ and
$\dsize\,\:\Te^{\:-1}\<(\<\ZZ\:;H)$ are Laurent polynomials in $\,\:\ZZ,H$
with integer coefficients.
\end{prop}
\begin{proof}
The first statement follows from formula \eqref{Sch2} and Corollary \ref{lemfV}.
For the second statement, we use in addition formulae \eqref{U}\,--\,\eqref{Weo}
and \eqref{tauXI}\:.
\end{proof}

\vsk.2>
\noindent
Set $\,\:\Med(\zz)=\Me\bigl(e^{\:2\pii\,z_1}\?\lc e^{\:2\pii\,z_n}\:\bigr)\,$
and $\,\:\Ted(\zz\:;h)\,=\,\Te\bigl(e^{\:2\pii\,z_1}\?\lc e^{\:2\pii\,z_n}\:;
e^{\:2\pii\,h}\:\bigr)\,$.

\begin{thm}
\label{monoint}
For any $\,I\?\in\Ikn\>$, we have
\vvn.3>
\begin{gather*}
\muo\Pso_{\?V_{\?I}\<(\<\GG\<)}\>=\<
\sum_{J\in\>\Ikn\!}\?\Med_{\IJ}(\<\zz)\,\Pso_{\?V_{\?J}\<(\<\GG\<)}\,,\kern1.4em
\muf\Pso_{\?V_{\?I}\<(\<\GG\<)}\>=\,e^{\:2\pii\,kh}\!\sum_{J\in\>\Ikn\!}\?
(\Ted\>\Med\>\Ted^{\:-1}\:)_{\IJ}(\<\zz\:;h)\,\Pso_{\?V_{\?J}\<(\<\GG\<)}\,,
\kern-.2em
\\[6pt]
\muo\Psf_{\?V_{\?I}\<(\<\GG\<)}\>=\<
\sum_{J\in\>\Ikn\!}\?(\Ted^{\:-1}\<\Med\>\Ted\>)_{\IJ}(\<\zz\:;h)\,
\Psf_{\?V_{\?J}\<(\<\GG\<)}\,,\kern1.4em
\muf\Psf_{\?V_{\?I}\<(\<\GG\<)}\>=\,e^{\:2\pii\,kh}\!
\sum_{J\in\>\Ikn\!}\?\Med_{\IJ}(\<\zz)\,\Psf_{\?V_{\?J}\<(\<\GG\<)}\,.\kern-.2em
\\[-16pt]
\end{gather*}
In particular, the entries of the monodromy matrices are Laurent polynomials in
\vv.07>
the variables $\,\:e^{\:2\pii\,z_1}\?\lc e^{\:2\pii\,z_n}\:;e^{\:2\pii\,h}\>$
with integer coefficients.
\end{thm}
\begin{proof}
The statement follows from Theorems \ref{monoP}, \ref{trans},
formulae \eqref{MeTe}\:, and Proposition \ref{lemMeTe}.
\end{proof}

\section{Integral formula for solutions of quantum differential equation}
\label{secint}
In this section we present integral formulae for solutions of
differential equation \eqref{mqde}\:. We use the integral formulae
to prove Theorem \ref{trans} in Section \ref{secpf}.

\subsection{Integral formula for solutions}
\label{secintf}
Recall the function $\,E(\Tb;H)\,$ given by \eqref{ET}\:,
\vvn.07>
the function $\,\:U(\Tb;\ZZ;H)\,$, see \eqref{U}\:, the reduced weight
\vv.07>
function $\,W\<(\TT\:;\gmb\:;\gmbb\:;h)\,$, see \eqref{W}\:, and the reduced
master function $\,\:\Phr\<(\TT\:;\zz\:;h)\,$, see \eqref{Phr}\:.
For any function $\,F(\Tb;\ZZ;H)\,$, set
\vvn.25>
\beq
\label{Fdd}
\Fdd(\TT\:;\zz\:;h)\,=\,F\bigl(e^{\:2\pii\,t_1}\?\lc e^{\:2\pii\,t_k}\:;
e^{\:2\pii\,z_1}\?\lc e^{\:2\pii\,z_n}\:;e^{\:2\pii\,h}\:\bigr)\,.\kern-2.5em
\vv.2>
\eeq
Denote
\vvn-.4>
\beq
\label{Ue}
\Ue(\TT\:;\zz\:;h)\,=\,\Udd(\TT\:;\zz\:;h)\,
\prod_{i=1}^k\:\prod_{a=k+1}^n\?\bigl(1-e^{\:2\pii\,(z_a\<-\:t_i)}\bigr)\,.
\kern-2em
\eeq

\vsk.1>
Let $\,\Lin\!\subset\C^n\!\times\C\,$ be the
following domain,
\vvn.3>
\beq
\label{Lin}
\Lin=\,\{(\zz,h)\ |\ \Re\:h\<<0\,,\;\;\Re\:(z_i\<-z_{i+1}\<+h)>0\,,\,\;
i=1\lc n-1\:\}\,.\kern-1em
\vv.3>
\eeq
For $\,(\zz,h)\<\in\Lin\<$ and $\,\:p\in\CRp\,$, consider the integral
\vvn.4>
\beq
\label{Ie}
\Ie(\zz\:;h\:;p)\,=\,\frac1{\bigl(2\:\piit\>\bigr)^{\<k}}
\int\limits_{\Ibb(\zz;h)}\!\!\<p^{\>\sum_{j=1}^{\:k}\?t_j}\>
\frac{\Ue(\TT\:;\zz\:;h)\,\Phr\<(\TT\:;\zz\:;h)
\,W\<(\TT\:;\gmb\:;\gmbb\:;h)}{\Edd(\TT\:;h)}\;d^{\:k}\TT\,.\kern-.5em
\eeq
Here the branch of $\,\:p^{\>\sum_{j=1}^{\:k}\?t_j}$ is fixed by
\vv.07>
the inequalities $\,-\:2\:\pi\<<\<\arg\:p<0\,$, and the integration contour
$\,\:\Ibb(\zz\:;h)\,$ is such that $\,\:{\Re\:t_i\<=\Re\:(z_i\<+h/2)}\,$
and $\,\:\Im\:t_i\>$ runs from $\,-\:\infty\,$ to $\,\infty\,$ for all
$\,i=1\lc k\,$.

\begin{prop}
\label{conv}
For $\,(\zz,h)\<\in\Lin\?$ and $\,-\:2\:\pi\<<\<\arg\:p<0\,$,
\vv.06>
the integrand in formula \eqref{Ie} is regular on
$\,\:\Ibb(\zz\:;h)\:$ and the integral $\,\:\Ie(\zz\:;h\:;p)\:$ converges
to a holomorphic function of $\,\zz,h,p\,$.
\end{prop}
\begin{proof}
By formulae \eqref{ET}\:, \eqref{U}\:, \eqref{Phr}\:,
the integrand in \eqref{Ie} equals
\vvn.2>
\begin{align}
\label{Igd}
C(\zz\:;h)\;p^{\>\sum_{j=1}^{\:k}\?t_j}\>e^{\:\pii\>\sum_{j=1}^{\:k}\?t_j}
\,\:W\<(\TT\:;\gmb\:;\gmbb\:;h)\,\:
\prod_{i=1}^k\,\biggl(\:\Gm(z_i\<-t_i)\,\Gm(t_i\<-z_i\<-h)\,\times{}\<&
\kern-1em
\\[2pt]
\notag
{}\times\,\prod_{j=1}^{i-1}\>
\frac{(t_i\<-t_j)\,\Gm(t_j\<-t_i\<+h)\,\Gm(z_j\<-t_i)}
{\Gm(1+t_j\<-t_i\<-h)\,\Gm(1+z_j\<-t_i\<+h)}
\,\prod_{j=i+1}^n\:\frac{\Gm(t_i\<-z_j\<-h)}{\Gm(1+t_i\<-z_j)}
\>\biggr)\>&,\kern-1em
\\[-15pt]
\notag
\end{align}
where $\,\:C(\zz\:;h)\,$ does not depend on $\,\:\TT\:,p\,$.
The poles of the integrand are at the hyperplanes
\vvn.4>
\beq
\label{poles}
t_i\<-z_j\<\in\Z_{\ge 0}\,,\;\;\, i\ge j\,,\qquad\!
t_i\<-z_j\<-h\in\Z_{\le 0}\,,\;\;\, i\le j\,,\qquad\!
t_i\<-t_j\<-h\in\Z_{\ge 0}\,,\;\;\, i>j\,.\kern-.8em
\vv.3>
\eeq
Hence, the integrand is regular on $\,\:\Ibb(\zz,h)\,$ under the assumption
$\,(\zz,h)\<\in\Lin\:$.
An estimate of the integrand
\vv.02>
showing the convergence of the integral follows from the inequality
\vvn.4>
\beq
\label{ineq}
\biggl|\>\frac{\Gm\bigl(\al+x\>\sqrt{\<-1}\>\bigr)}
{\Gm\bigl(\:\bt+x\>\sqrt{\<-1}\>\bigr)}\:\biggr|\,\le\,
C_1\:\bigl(1+|\:x\:|\:\bigr)^{\Re\:(\al-\bt)}\kern-2em
\vv.2>
\eeq
that hold for any real $\,\:x\,$, fixed $\,\al\:,\:\bt\,$ such that
$\,\:\Re\bt\not\in\<\Z_{\le 0}\,$,
and some positive constant $\,\:C_1\>$ depending on $\,\al\:,\:\bt\,$.
\end{proof}

Recall the class $\,\:Q(\gmb\:;\gmbb\:;h)\,$ given by \eqref{QG}\:.
Similarly to \eqref{Psof}\:, define
\vvn.4>
\beq
\label{Psin}
\Psin\<(\zz\:;h\:;p)\,=\,
\Ie(\zz\:;h\:;p)\,Q(\gmb\:;\gmbb\:;h)\,,\kern-1.9em
\eeq

\vsk.4>
Recall the trigonometric weight function $\,\:\Weo\<(\Tb;\ZZ;H)\,$,
\vvn.06>
see \eqref{Weo}\:, and the maps $\,\:\pio\>$ and $\,\:\pif\:$ given
by \eqref{pimap}\:. For a permutation $\,\:\si\<\in\<S_n\,$,
\vvn.4>
define the classes $\,\:\Yo_{\<\si},\:\Yf_{\<\si}\!\<\in\<\KT\,$,
\beq
\label{Yofpi}
\Yo_{\<\si}\?=\>\pio\:\bigl(\Weo\<(\Tb;\ZZ_{\?\si};H)\<\bigr)\,,\qquad
\Yf_{\<\si}\?=\>\pif\:\bigl(\Weo\<(\Tb;\ZZ_{\?\si};H)\<\bigr)\,.\kern-2em
\vv.4>
\eeq
Consider the solutions $\,\:\Pso_{\?Y^0_{\<\si}}(\zz\:;h\:;p)\,$ and
$\,\Psf_{\<Y^{\<\infty}_{\<\si}\vp1}\<(\zz\:;h\:;p)\,$ of differential
equation \eqref{mqde} given by \eqref{PsofX}\:. We regard them as functions
\vv.04>
of $\,p\in\CRp\,$ by identifying $\,\:\CRp\,$ and its image $\;\BB\>$
in the universal cover of $\,\:\Czon\,$, see Section \ref{secmono}.

\begin{thm}
\label{thmPsin}
Given $\,\si\<\in\<S_n\,$, let $\,(\zz_\si,h)\<\in\Lin\?$.
Assume that $\,\:p\in\CRp\,$. Then
\vvn.3>
\beq
\label{Psin08}
\Psin(\zz_\si;h\:;p)\,=\,\bigl(\:\Gm(\<-\:h)\<\bigr)^k
\>\Pso_{\?Y^0_{\<\si}}(\zz\:;h\:;p)\;\;\,\text{and}\,\;\;\,
\Psin(\zz_\si\:;h\:;p)\,=\,\bigl(\:\Gm(\<-\:h)\<\bigr)^k
\>\Psf_{\<Y^{\<\infty}_{\<\si\vp1}}\<(\zz\:;h\:;p)\,.\kern-.8em
\vv.1>
\eeq
\end{thm}
\vsk.2>
\noindent
Theorem \ref{thmPsin} is proved in the next section.

\subsection{Proof of Theorem \ref{thmPsin}}
\label{secpfth}
It suffices to prove the statement for the identity permutation
$\,\si\,$. To simplify writing, we will omit $\,\sigma\,$ from the notation.
We will prove the equality
\vvn.3>
\beq
\label{Psin0}
\Psin(\zz\:;h\:;p)\,=\,\bigl(\:\Gm(\<-\:h)\<\bigr)^k
\>\Pso_{\?Y^0}(\zz\:;h\:;p)\kern-2em
\vv.1>
\eeq
for $\,\:|\:p\:|\<<1\,$ and the equality
\beq
\label{Psin8}
\Psin(\zz\:;h\:;p)\,=\,\bigl(\:\Gm(\<-\:h)\<\bigr)^k
\>\Psf_{\<Y^{\<\infty}_{\vp1}}\<(\zz\:;h\:;p)\kern-2em
\vv.1>
\eeq
for $\,\:|\:p\:|>1\,$. Then for an arbitrary $\,p\,$, these equalities
will hold by analytic continuation.

\vsk.3>
Let $\,\:|\:p\:|\<<1\,$. Set $\,\zzst\?=(z_1\lc z_k)\,$. Recall
\vvn.4>
\begin{align}
\label{PsioR}
& \!\<\Pso_{\?Y^0}(\zz\:;h\:;p)\,={}
\\[5pt]
\notag
\kern-4em\!\<{}=\,Q(\gmb &{}\:;\gmbb\:;h)
\sum_{I\in\>\Ikn\!}\!\biggl(\>\frac{\Weod\<(\zz_I;\zz\:;h)}{\Edd(\zz_I;h)}
\>\sum_{\lb\in\:\Z_{\ge 0}^k\!}\Res_{\>\TT\:=\zz_I\<+\:\lb}
\bigl(\Phr\<(\TT\:;\zz\:;h)\>W\<(\TT\:;\gmb\:;\gmbb\:;h)\<\bigr)\,
p^{\>\sum_{i=1}^{\:k}\<\<(\<z_i\<+\:l_i\<)}\biggr)\:,\kern-1.8em
\\[-16pt]
\notag
\end{align}
see \eqref{Yofpi}\:, \eqref{pimap}\:, \eqref{PsofX}\:, \eqref{Psof}\:,
\eqref{Fo}\:. By Lemma \ref{lemWeIJ}, the sum over $\>I\?\in\Ikn\>$
\vv.07>
in \eqref{PsioR} reduces to a single term when $\>I=\onek\,$ and
\vv.08>
$\,\zz_I\<=\zzst$. Then using formulae \eqref{WeIJ}\:, \eqref{RZ}\:,
and the definition of $\,\:\Psin(\zz\:;h\:;p)\,$, see \eqref{Psin}\:,
we convert equality \eqref{Psin0} to the following form,
\vvn.2>
\begin{align}
\label{IeRes}
\Ie(\zz\:;h\:;p)\,&{}=\,(-1)^k\>\prod_{a=1}^k\:\prod_{b=k+1}^n\?
\bigl(1-e^{\:2\pii\,(z_b\<-\:z_a)}\bigr)\,\times{}
\\[4pt]
\notag
&{}\>\times\sum_{\lb\in\:\Z_{\ge 0}^k\!}\Res_{\>\TT\:=\zzss\!+\:\lb}
\bigl(\Phr\<(\TT\:;\zz\:;h)\>W\<(\TT\:;\gmb\:;\gmbb\:;h)\<\bigr)\,
p^{\>\sum_{i=1}^k\<(\<z_i\<+\:l_i\<)}\:,\kern-2em
\end{align}
where $\,\:\Ie(\zz\:;h\:;p)\,$ is integral \eqref{Ie}\:. To prove formula
\eqref{IeRes}\:, we evaluate the integral $\,\:\Ie(\zz\:;h\:;p)\,$
via the sum of residues ``\:to the right'' \<of the integration contour
$\,\:\Ibb(\zz\:;h)\,$.

\vsk.3>
Consider the $\>p\>$-independent part of the integrand in \eqref{Ie}\:,
\vvn.4>
\beq
\label{F}
F(\TT\:;\zz\:;h)\,=\,\frac{\Ue(\TT\:;\zz\:;h)\,\Phr\<(\TT\:;\zz\:;h)
\,W\<(\TT\:;\gmb\:;\gmbb\:;h)}{\Edd(\TT\:;h)}\;.
\vv.3>
\eeq

\begin{lem}
\label{lemF1}
For $\,\:\lb\in\<\Z_{\ge 0}\,$, we have
\vvn.4>
\begin{align}
\label{ResFk}
& \Res_{\>t_k=\:z_k\<+\>l_k}\ldots\>\Res_{\>t_1=\:z_1\<+\>l_1}\?
F(\TT\:;\zz\:;h)\,={}
\\[-1pt]
\notag
&\;{}=\,\Res_{\>\TT\:=\zzss\!+\:\lb}
\bigl(\Phr\<(\TT\:;\zz\:;h)\>W\<(\TT\:;\gmb\:;\gmbb\:;h)\<\bigr)\,
\prod_{a=1}^k\:\prod_{b=k+1}^n\?\bigl(1-e^{\:2\pii\,(z_b\<-\:z_a)}\bigr)\,.
\kern-2em
\\[-16pt]
\notag
\end{align}
\end{lem}
\begin{proof}
The functions $\,\:\Ue(\TT\:;\zz\:;h)/\Edd(\TT\:;h)\,$ and
\vv.1>
$\,\:W\<(\TT\:;\gmb\:;\gmbb\:;h)\,$ are holomorphic in a neighbourhood
of the locus $\,\:\TT=\zzst\!\<+\lb\,$, and the function
$\,\:\Phr\<(\TT\:;\zz\:;h)\,$ has the form
\vvn.26>
\be
\Phr\<(\TT\:;\zz\:;h)\,=\,
G(\TT\:;\zz\:;h)\alb\:\prod_{\:i=1}^{\:k}\Gm(z_i\<-t_i)\,,\kern-1.4em
\vv.3>
\ee
where the function $\,G(\TT\:;\zz\:;h)\,$ is holomorphic in a neighbourhood
of $\,\:\TT=\zzst\!\<+\lb\,$. Therefore,
\vvn.6>
\begin{align*}
\Res_{\>t_k=\:z_k\<+\>l_k}\ldots\>\Res_{\>t_1=\:z_1\<+\>l_1}\?
F(\TT\:;\zz\:;h)\,=\,
\Res_{\>t_1=\:z_1\<+\>l_1}\ldots\>\Res_{\>t_k=\:z_k\<+\>l_k}\?
F(\TT\:;\zz\:;h)\,={} \kern-.8em &
\\[4pt]
{}=\,\:\frac{\Ue(\zzst\!\<+\lb\:;\zz\:;h)}{\Edd(\zzst\!\<+\lb\:;h)}
\,\Res_{\>\TT\:=\zzss\!+\:\lb}
\bigl(\Phr\<(\TT\:;\zz\:;h)\>W\<(\TT\:;\gmb\:;\gmbb\:;h)\<\bigr)
\kern-.82em &
\\[-14pt]
\end{align*}
by definition \eqref{Res} of $\,\:\Res_{\>\TT\:=\zzss\!+\:\lb}\>$.
Since
\vvn.3>
\be
\frac{\Ue(\zzst\!\<+\lb\:;\zz\:;h)}{\Edd(\zzst\!\<+\lb\:;h)}\,\:=
\,\prod_{a=1}^k\:\prod_{b=k+1}^n\?\bigl(1-e^{\:2\pii\,(z_b\<-\:z_a)}\bigr)\,,
\kern-2em
\vv-.1>
\ee
formula \eqref{ResFk} follows.
\end{proof}

For $\,m\le k\,$, denote $\,\:\TT^{\ges m}\!=(t_{m+1}\lc t_k)\,$.
The sequence $\,\:\TT^{\ges k}$ is empty and will be used for convenience.
For nonnegative integers $\,\:l_1\lc l_m\>$, set
\vvn.4>
\beq
\label{FlmR}
F_{\:l_1\?\lc\:l_m}\?(\TT^{\ges m};\zz\:;h)\,=\,\Res_{\>t_m=\:z_m\<+\>l_m}\ldots
\>\Res_{\>t_1=\:z_1\<+\>l_1}\?F(\TT\:;\zz\:;h)\,.\kern-2em
\vv.4>
\eeq
The function $\,F_{\:l_1\?\lc\:l_k}\?(\TT^{\ges k};\zz\:;h)\>$ does not depend
\vv.08>
on $\,\:\TT\,$ and we can omit the argument $\,\:\TT^{\ges k}\<$.
Then formula \eqref{ResFk} reads
\vvn-.5>
\be
F_{\:l_1\?\lc\:l_k}\?(\zz\:;h)\,=\,\Res_{\>\TT\:=\zzss\!+\:\lb}
\bigl(\Phr\<(\TT\:;\zz\:;h)\>W\<(\TT\:;\gmb\:;\gmbb\:;h)\<\bigr)\,
\prod_{a=1}^k\:\prod_{b=k+1}^n\?\bigl(1-e^{\:2\pii\,(z_b\<-\:z_a)}\bigr)\,,
\kern-2em
\vv-.4>
\ee
and formula \eqref{IeRes} takes the form
\vvn-.1>
\beq
\label{IeResF}
\Ie(\zz\:;h\:;p)\,=\,(-1)^k\!\<\sum_{l_1\?\lc\:l_k=\:0\!\!}^\infty\!\?
F_{\:l_1\?\lc\:l_k}\?(\zz\:;h)\,p^{\>\sum_{i=1}^k\<(\<z_i\<+\:l_i\<)}\:.
\kern-2em
\eeq

\vsk.2>
In what follows, we will prove formula \eqref{IeResF}\:.
\vv.04>
We start the proof with technical lemmas.
The main part of the proof is given by Proposition \ref{lemIeResF}.
Formula \eqref{IeResF} coincides with formula \eqref{Iems} for $\,m=k\,$.

\begin{lem}
\label{lemF2}
For nonnegative integers $\,\:l_1\lc l_{m+1}\>$, $m<k\,$,
\vvn.5>
and any $\,\:i=1\lc m\,$, we have
\be
\Res_{\>t_{m+1}=\:z_i+\>l_{m+1}}\?F_{\:l_1\?\lc\:l_m}\?(\TT^{\ges m};\zz\:;h)
\,=\:-\:\Res_{\>t_{m+1}=\:z_i+\>l_i}\<
F_{\:l_1\?\lc\:l_{i-1}\<,\>l_{m+1}\?,\>l_{i+1}\?\lc\:l_m}\?
(\TT^{\ges m};\zz\:;h)\,.\kern-1em
\vv.5>
\ee
\end{lem}
\begin{proof}
For each $\,\:i=1\lc m\,$, the function $\,F(\TT\:;\zz\:;h)\,$ as a function
of $\,\:t_1\lc\:t_{m+1}\>$ has the form
\vvn-.2>
\be
G_i(\TT\:;\zz\:;h)\,H_i(\TT\:;\zz\:;h)\,e^{\:\pii\>(t_i-\:t_{m+1}\<)}
(t_{m+1}\<-t_i)\:\,\prod_{j=1}^m\,\Gm(z_j\<-t_j)\,\Gm(z_j\<-t_{m+1})\,,
\kern-2em
\vv.4>
\ee
where the function $\,G_i(\TT\:;\zz\:;h)\,$ is symmetric in
$\,\:t_i\:,\:t_{m+1}\,$, holomorphic in a neighbourhood of the point
\vvn-.2>
\be
t_{m+1}\<=z_i\<+l_{m+1}\,,\qquad t_j\<=z_j\<+l_j\,\quad j=1\lc m\,,\kern-1em
\vv.5>
\ee
and holomorphic in a neighbourhood of the point
\vvn.4>
\be
t_{m+1}\<=z_i\<+l_i\,,\quad t_i\<=z_i\<+l_{m+1}\,,\qquad
t_j\<=z_j\<+l_j\,\quad j=1\lc i-1,i+1\lc m\,,
\vv.4>
\ee
while $\,H_i(\TT\:;\zz\:;h)\,$ is a $\,1$-\:periodic function of
\vv.04>
$\,\:t_1\lc\:t_{m+1}\>$ holomorphic in a neighbourhood of the point
$\,\:t_{m+1}\<=z_i\,$, $\,\:t_j=z_j\,$, $\,j=1\lc m\,$.
This implies the statement of Lemma \ref{lemF2}.
\end{proof}

\begin{lem}
\label{Wth}
For any $\,\:i,j=1\lc k\,$, $\,\:i\ne j\,$, and any $\,\:a=1\lc n\,$, we have
\vvn.5>
\be
W\<(\TT\:;\gmb\:;\gmbb\:;h)|_{\:t_i=\:z_a\<,\,t_j=\:z_a\?+\:h}^{}\>=\,0
\vv-.6>
\ee
in $\,\HT\,$.
\end{lem}
\begin{proof}
Formula \eqref{W} implies that the expression
$\,W\<(\TT\:;\gmb\:;\gmbb\:;h)|_{\:t_i=\:z_a\<,\,t_j=\:z_a\?+\:h}^{}\>$
is divisible by the product
$\,\:\prod_{\:i=1}^{\:k}(z_a\<-\gm_i)\:\prod_{\:j=1}^{\:n-k}(z_a\<-\gmbr_j)\,$.
\vv.1>
This product equals zero in $\,\HT\,$ according to relations \eqref{Hrel}\:.
Hence, Lemma \ref{Wth} follows.
\end{proof}

\begin{lem}
\label{lemF3}
For nonnegative integers $\,\:l_1\lc l_m\>$, $\,m<k\,$, the function
\vv.08>
$\,F_{\:l_1\?\lc\:l_m}\?(\TT^{\ges m};\zz\:;h)\>$ is regular at the hyperplanes
$\,\:t_{m+1}\<-z_i\<-l_i\<-h\<\in\<\Z_{\ge0}\,$, $\,\:i=1\lc m\,$.
\end{lem}
\begin{proof}
By formulae \eqref{Igd}\:, \eqref{FlmR}\:,
\vvn.4>
\begin{alignat}2
\label{Flm}
& \;F_{\:l_1\?\lc\:l_m}\?(\TT^{\ges m};\zz\:;h)\,={}
\\[9pt]
\notag
& {}=\,C_{l_1\?\lc\:l_m}\?(\zz\:;h)\;e^{\:\pii\>\sum_{j=m+1}^{\:k}\?t_j}
\,\:W\<(\TT\:;\gmb\:;\gmbb\:;h)|
_{\:t_1=\:z_1\<+\>l_1\<\lc\>t_m=\:z_m\<+\>l_m}^{}\times{}
\\[8pt]
\notag
& \:{}\times\<\prod_{i=m+1}^k\<\biggl(\:
\Gm(z_i\<-t_i)\,\Gm(t_i\<-z_i\<-h)\;\prod_{a=1}^m\,
\frac{(t_i\<-z_a\?-l_a)\,\Gm(z_a\?-t_i)\,\Gm(z_a\?+l_a\?-t_i\<+h)}
{\Gm(1+z_a\?+l_a\?-t_i\<-h)\,\Gm(1+z_a\?-t_i\<+h)}\,\times{}\kern-.4em &&
\\[4pt]
\notag
&& \llap{$\dsize{}\times\<\prod_{j=m+1}^{i-1}
\frac{(t_i\<-t_j)\,\Gm(t_j\<-t_i\<+h)\,\Gm(z_j\<-t_i)}
{\Gm(1+t_j\<-t_i\<-h)\,\Gm(1+z_j\<-t_i\<+h)}
\,\prod_{j=i+1}^n\:\frac{\Gm(t_i\<-z_j\<-h)}{\Gm(1+t_i\<-z_j)}
\>\biggr)\>$},\!\kern-.4em &
\\[-12pt]
\notag
\end{alignat}
where $\,\:C_{l_1\?\lc\:l_m}\?(\zz\:;h)\,$ does not depend on $\,\:\TT\,$.
Therefore, $\,F_{\:l_1\?\lc\:l_m}\?(\TT\:;\zz\:;h)\,$ has the form
\vvn.5>
\be
G(\TT\:;\zz\:;h)\,W\<(\TT\:;\gmb\:;\gmbb\:;h)|
_{\:t_1=\:z_1\<+\>l_1\<\lc\>t_m=\:z_m\<+\>l_m}^{}\,\:\prod_{j=1}^m\,
\frac{\Gm(z_j\<+l_j\<-t_{m+1}\<+h)}{\Gm(1+z_j\<-t_{m+1}\<+h)}\;,\kern-1.8em
\vv.3>
\ee
where $\,G(\TT\:;\zz\:;h)\,$ as a function of $\,\:t_{m+1}\>$ is holomorphic
\vv.07>
in a neighbourhood of the point $\,\:t_{m+1}\<=z_i\?+l_i\<+r+h\,\:$
for any $\,\:i=1\lc m\,$, \:and $\>r\?\in\<\Z_{\ge0}\,$. The ratio of
Gamma\:-\:functions
\vvn.5>
\be
\frac{\Gm(z_i\<+l_i\<-t_{m+1}\<+h)}{\Gm(1+z_i\<-t_{m+1}\<+h)}\kern-2em
\vv.3>
\ee
is a polynomial if $\,\:l_i\ge 1\,$, and equals $\,(z_i\<-t_{m+1}\<+h)^{-1}\:$
\vv.04>
if $\,\:l_i\<=0\,$. In the last case, the function
$\,F_{\:l_1\?\lc\:l_m}\?(\TT\:;\zz\:;h)\,$ does not have
a pole at $\,\:t_{m+1}\<=z_i\?+h\,$ due to Lemma \ref{Wth}.
\vv.04>
Lemma \ref{lemF3} is proved.
\end{proof}

Recall the integral $\,\:\Ie(\zz\:;h\:;p)\,$, see \eqref{Ie}\:,
and the functions $\,F_{\:l_1\?\lc\:l_m}\?(\TT^{\ges m};\zz\:;h)\,$,
see \eqref{FlmR}\:.

\begin{prop}
\label{lemIeResF}
Let $\,\:|\:p\:|\<<1\,$. For any $\,m=1\lc k\,$, we have
\vvn.4>
\beq
\label{Iems}
\Ie(\zz\:;h\:;p)\,=\,(-1)^m\!\<\sum_{l_1\?\lc\:l_m=\:0\!\!}^\infty\!
\Ie_{l_1\?\lc\:l_m}\?(\zz\:;h\:;p)\,
p^{\>\sum_{i=1}^{\:m}\<(\<z_i\<+\:l_i\<)}\:,\kern-2em
\vv.3>
\eeq
where
$\,\:\Ie_{l_1\?\lc\:l_k}\?(\zz\:;h\:;p)=F_{\:l_1\?\lc\:l_k}\?(\zz\:;h)\:$
does not depend on $\,p\,$, and for $\,m<k\,$,
\vvn.4>
\beq
\label{Iem}
\Ie_{l_1\?\lc\:l_m}\?(\zz\:;h\:;p)\,=\,\frac1{\bigl(2\:\piit\>\bigr)^{\<k-m}}
\int\limits_{\Ibb_m\<(\<\zz;h)}\!\!\<p^{\>\sum_{j=m+1}^{\:k}\?t_j}\:
F_{\:l_1\?\lc\:l_m}\?(\TT^{\ges m};\zz\:;h)\;d\:t_{m+1}\:\ldots d\:t_k\,,
\kern-.66em
\eeq
the integration contour $\,\:\Ibb_{\:m}\<(\zz\:;h)\,$ being such that
\vv.08>
$\,\:{\Re\:t_i\<=\Re\:(z_i\<+h/2)}\,$ and $\,\,\Im\:t_i\>$ runs
from $\,-\:\infty\,$ to $\,\infty\,$ for all $\,i=m+1\lc k\,$.
\end{prop}
\begin{proof}
The statement is proved by induction on $\,m\,$. The base of induction at
\vv.04>
$\,m=1\,$ amounts to evaluating integral \eqref{Ie} with respect to $\,\:t_1\,$
\vv.06>
using the sum of residues in the half\:-plane
$\,\:{\Re\:t_1\?>\<\Re\:(z_1\<+h/2)}\,$. The relevant poles of the integrand
are of the form $\,\:t_1\<=z_1\?+l_1\>$, $\,\:l_1\?\in\Z_{\ge 0}\,$.

\vsk.2>
For the induction step, we evaluate integral \eqref{Iem} with respect
\vv.04>
to $\,\:t_{m+1}\,$ using the sum of residues in the half\:-plane
\vv.04>
$\,\:{\Re\:t_{m+1}\<>\Re\:(z_{m+1}\<+h/2)}\,$. By formula \eqref{Flm} and
Lemma \ref{lemF3}, the relevant poles of the integrand are of the form
\vv.04>
$\,\:t_{m+1}\<=z_i\?+l_{m+1}\>$, $\,i=1\lc m+1\,$,
$\,\:l_{m+1}\?\in\Z_{\ge 0}\,$. Plugging the obtained series for
$\,\:\Ie_{l_1\?\lc\:l_m}\?(\zz\:;h\:;p)\,$ into formula \eqref{Iems}
\vv.06>
and using formula \eqref{FlmR}\:, we get
\begin{align}
\label{Iemm}
\Ie(\zz\:;h\:;p)\,=\,{}& (-1)^{m+1}\!\<\sum_{l_1\?\lc\:l_{m+1}=\:0\!\!}^\infty
\!\Ie_{l_1\?\lc\:l_{m+1}}\?(\zz\:;h\:;p)\,
p^{\>\sum_{i=1}^{\:m+1}\<(\<z_i\<+\:l_i\<)}\:+{}\kern-2em
\\[6pt]
\notag
{}+\,{}& (-1)^{m+1}\,\:\sum_{i=1}^m\>\sum_{l_1\?\lc\:l_{m+1}=\:0\!\!}^\infty
\!\Ie^{\:(i)}_{l_1\?\lc\:l_{m+1}}\?(\zz\:;h\:;p)\,
p^{\>z_i+\:l_{m\<+\<1}\<+\:\sum_{i=1}^{\:m}\<(\<z_i\<+\:l_i\<)}\kern-2em
\\[-26pt]
\notag
\end{align}
where
\vvn.4>
\begin{align*}
& \Ie^{\:(i)}_{l_1\?\lc\:l_{m+1}}\?(\zz\:;h\:;p)\,={}
\\[4pt]
& \!\?{}=\,\frac1{\bigl(2\:\piit\>\bigr)^{\<k-m-1}}
\int\limits_{\Ibb_{m+1}\<(\<\zz;h)}\!\!\<p^{\>\sum_{j=m+2}^{\:k}\?t_j}\:
\Res_{\>t_{m+1}=\:z_i+\>l_{m+1}}\?F_{\:l_1\?\lc\:l_m}\?(\TT^{\ges m};\zz\:;h)
\;d\:t_{m+2}\:\ldots d\:t_k\,.\kern-.8em
\\[-14pt]
\end{align*}
The second term in the right-hand side of formula \eqref{Iemm} equals zero
because
\vvn.5>
\be
\Ie^{\:(i)}_{l_1\?\lc\:l_{m+1}}\?(\zz\:;h\:;p)\,=\:-\,\:
\Ie^{\:(i)}_{\:l_1\?\lc\:l_{i-1}\<,\>l_{m+1}\?,\>l_{i+1}\?\lc\:l_m\?,\>l_i}\?
(\zz\:;h\:;p)\kern-2em
\vv.5>
\ee
by Lemma \ref{lemF2}. This completes the induction step.
Proposition \ref{lemIeResF} is proved.
\end{proof}

Formula \eqref{Iems} for $\,m=k\,$ coincides with formula \eqref{IeResF}\:.
This completes the proof of formula \eqref{IeResF} and thus the proof of
equality \eqref{Psin0} for $\,\:|\:p\:|\<<1\,$.

\vsk.4>
Let $\,\:|\:p\:|>1\,$. Consider equality \eqref{Psin8}\:.
We convert it to the following form,
\vvn.3>
\begin{align}
\label{IeResh}
& \Ie(\zz\:;h\:;p)\,={}
\\[5pt]
\notag
\llap{${}={}$}\?\kern1.5em & \<\kern-1.5em \sum_{I\in\>\Ikn\!}\?
\frac{\Weod\<(\zz_I\<+\hb\:;\zz\:;h)}{\Edd(\zz_I;h)}\,
\sum_{\lb\in\:\Z_{\ge 0}^k\!}\Res_{\>\TT\:=\zz_I\<+\hb-\:\lb}
\bigl(\Phr\<(\TT\:;\zz\:;h)\>W\<(\TT\:;\gmb\:;\gmbb\:;h)\<\bigr)\,
p^{\>kh\:+\sum_{i\in I}\<z_i-\sum_{i=1}^{\:k}\?l_i}\:.\kern-.46em
\\[-16pt]
\notag
\end{align}

To this end, we use the definition of $\,\:\Psin(\zz\:;h\:;p)\,$,
see \eqref{Psin}\:, and the definition of
$\,\:\Psf_{\?Y^{\<\infty}_{\vp1}}\<(\zz\:;h\:;p)\,$ given by \eqref{Yofpi}\:,
\eqref{pimap}\:, \eqref{PsofX}\:, \eqref{Psof}\:, \eqref{Ff}\:.
To prove formula \eqref{IeRes}\:, we evaluate the integral
$\,\:\Ie(\zz\:;h\:;p)\,$ via the sum of residues ``\:to the left''
\<of the integration contour $\,\:\Ibb(\zz\:;h)\,$.

\vsk.3>
For $\,\:\si\<\in\<S_n\,$, denote $\,I_\si\<=\{\:\si(1)\lc\si(k)\:\}\<\in\Ikn\,$
\vv.07>
and $\,\:\zzs_{\<\si}\<+\hb=(z_{\si(1)}\?+h\lc z_{\si(k)}\?+h)\,$.
Denote by $\,S_n[\:k\:]\,$ the set of $\,\:\si\in S_n\>$ such that
\vv.07>
$\,\:\si(i)\ge i\>$ for all $\,\:i\le k\,$, and $\,\:\si(i)>\si(j)\,\:$
for all $\,i>j>k\,$.

\vsk.2>
Recall the function $\,\:\Ue(\TT\:;\zz\:;h)\,$, see \eqref{Ue}\:.
\vv.08>
By inspection of formulae \eqref{Ue}\:, \eqref{U}\:, we have
$\,\:\Ue(\zzs_{\<\si}\<+\hb\:;\zz\:;h)=0\,$ unless $\,\:\si(i)\ge i\>$
\vv.08>
for all $\,\:i\le k\,$. Clearly, $\,\:\zzs_{\<\si}\<+\hb\,$ does not depend on
the values $\,\:\si(j)\,$ for $\,j>k\,$ Thus formulae \eqref{We}\:, \eqref{Weo}
yield
\vvn.3>
\beq
\label{WUsi}
\Weod\<(\zz_I\<+\hb\:;\zz\:;h)\,=\!
\sum_{\satop{\si\in S_n\<[\:k\:]\!\!\!\!}{I=I_\si}}
\Ue(\zzs_{\<\si}\<+\hb\:;\zz\:;h)\,.\kern-2em
\eeq

\vsk.1>
The rest of the proof of equality \eqref{Psin8} for $\,\:|\:p\:|>1\,$
\vv.07>
is analogous to that of equality \eqref{Psin0} for $\,\:|\:p\:|\<<1\,$, but
involves more combinatorics due to the sum over $\,\:S_n[\:k\:]\,$ in
formula \eqref{WUsi}\:.

\vsk.3>
Recall the $\>p\>$-independent part $\,F(\TT\:;\zz\:;h)\,$
\vv.07>
of the integrand in \eqref{Ie}\:, see \eqref{F}\:. For $\,m\le k\,$, denote
$\,\:\TT^{\les m}\!=(t_1\lc t_{m-1})\,$. The sequence $\,\:\TT^{\les 1}$
\vv.06>
is empty and will be used for convenience. For $\,\:\si\<\in\<S_n\,$
and nonnegative integers $\,\:l_m\lc l_k\>$, set
\vvn.5>
\beq
\label{FlmRh}
F^{\:\si}_{l_m\?\lc\:l_k}\?(\TT^{\les m};\zz\:;h)\,=\,
\Res_{\>t_m=\:z_{\si\<(\?m\?)}\<+\:h-\>l_m}\ldots\>
\Res_{\>t_k=\:z_{\si\<(\<k\<)}\<+\:h-\>l_k}\?F(\TT\:;\zz\:;h)\,.\kern-1em
\vv.5>
\eeq
The function $\,F^{\>\si}_{l_1\?\lc\:l_k}\?(\TT^{\les 0};\zz\:;h)\>$ does not
\vv.06>
depend on $\,\:\TT\,$ and we can omit the argument $\,\:\TT^{\les 0}\<$.

\begin{lem}
\label{lemFh0}
For $\,\:\si\<\in\<S_n\:$ and $\;\lb\in\<\Z_{\ge 0}\,$, we have
\vvn.5>
\be
F^{\:\si}_{l_1\?\lc\:l_k}\?(\zz\:;h)\,=\,
\frac{\Ue(\zzs_{\<\si}\<+\hb\:;\zz\:;h)}{\Edd(\zzs_{\<\si}\<+\hb\:;h)}\,
\Res_{\>\TT\:=\zzss_{\?\si}\?+\hb-\:\lb}
\bigl(\Phr\<(\TT\:;\zz\:;h)\>W\<(\TT\:;\gmb\:;\gmbb\:;h)\<\bigr)\,.\kern-.08em
\vv.3>
\ee
\end{lem}
\begin{proof}
The functions $\,\:\Ue(\TT\:;\zz\:;h)/\Edd(\TT\:;h)\,$ and
\vv.06>
$\,\:W\<(\TT\:;\gmb\:;\gmbb\:;h)\,$ are holomorphic in a neighbourhood
of the locus $\,\:\TT=\zzs_{\<\si}\<+\hb-\lb\,$, and the function
$\,\:\Phr\<(\TT\:;\zz\:;h)\,$ has the form
\vvn.2>
\be
\Phr\<(\TT\:;\zz\:;h)\,=\,G(\TT\:;\zz\:;h)\,\:
\prod_{\:i=1}^{\:k}\,\Gm(t_i\<-z_i\<-h)\,,\kern-1.4em
\vv.3>
\ee
where the function $\,G(\TT\:;\zz\:;h)\,$ is holomorphic in a neighbourhood
\vv.07>
of $\,\:\TT=\zzs_{\<\si}\<+\hb-\lb\,$. Furthermore, the function
$\,\:\Ue(\TT\:;\zz\:;h)/\Edd(\TT\:;h)\,$ is $\,1$-\:periodic in each of
$\,\:t_1\lc t_k\,$. Hence, the statement follows.
\end{proof}
Notice that $\,\:\Edd(\zzs_{\<\si}\<+\hb\:;h)=\Edd(\zz_{\?I_\si}^{}\<;h)\,$.
\vv.07>
Then by formula \eqref{WUsi} and Lemma \ref{lemFh0}, equality \eqref{IeResh}
takes the form
\vvn-.3>
\beq
\label{IeResFh}
\Ie(\zz\:;h\:;p)\,=\?\sum_{l_1\?\lc\:l_k=\:0\!\!}^\infty\;\>
\sum_{\si\in S_n\<[\:k\:]\!\!\!\!}\,F^{\:\si}_{l_1\?\lc\:l_k}\?(\zz\:;h)\,
p^{\>\sum_{i=1}^k\<(\<z_{\si\<(\<i\<)}\<+\:h-\:l_i\<)}\:.\kern-2em
\eeq

\vsk.4>
In what follows, we will prove formula \eqref{IeResFh}\:.
\vv.04>
We start the proof with technical lemmas.
The main part of the proof is given by Proposition \ref{lemIeResFh}.
Formula \eqref{IeResFh} coincides with formula \eqref{Iemsh} for $\,m=1\,$.

\begin{lem}
\label{lemFh2}
For $\,\:\si\<\in\<S_n\,$, nonnegative integers $\,\:l_{m-1}\lc l_k\>$,
$m\ge 2\,$, and any $\,\:i=m\lc k\,$, we have
\vvn.2>
\be
\Res_{\>t_{m-1}=\:z_{\si\<(\<i\<)}\<+\:h-\>l_{m-1}}\?
F^{\:\si}_{l_m\?\lc\:l_k}\?(\TT^{\les m};\zz\:;h)
\,=\:-\:\Res_{\>t_{m-1}=\:z_{\si\<(\<i\<)}\<+\:h-\>l_i}\<
F^{\:\si}_{l_m\?\lc\:l_{i-1}\<,\>l_{m-1}\?,\>l_{i+1}\?\lc\:l_k}\?
(\TT^{\les m};\zz\:;h)\,.\kern-.1em
\vv.5>
\ee
\end{lem}
\begin{proof}
For each $\,\:i=m\lc k\,$, the function $\,F(\TT\:;\zz\:;h)\,$ as a function
of $\,\:t_{m-1}\lc\:t_k\>$ has the form
\vvn-.2>
\be
\kern-.4em
G^{\:\si}_i\<(\TT\:;\zz\:;h)\,H^{\:\si}_i\<(\TT\:;\zz\:;h)\,
e^{\:\pii\>(t_{m-1}-\:t_i\<)}(t_i\<-t_{m-1})\:\,\prod_{j=m}^k\,
\Gm(t_j\<-z_{\si(j)}\<-h)\,\Gm(t_{m-1}\<-z_{\si(j)}\<-h)\,,\kern-.7em
\vv.2>
\ee
where the function $\,G^{\:\si}_i\<(\TT\:;\zz\:;h)\,$ is symmetric in
$\,\:t_{m-1}\:,\:t_i\,$, holomorphic in a neighbourhood of the point
\be
t_{m-1}\<=z_{\si(i)}\<+h-l_{m-1}\,,\qquad
t_j\<=z_{\si(j)}\<+h-l_j\,\quad j=m\lc k\,,\kern-1em
\vv.7>
\ee
and holomorphic in a neighbourhood of the point
\vvn.4>
\begin{gather*}
t_{m-1}\<=z_{\si(i)}\<+h-l_i\,,\qquad t_i\<=z_{\si(i)}\<+h-l_{m-1}\,,
\kern-1,2em
\\[4pt]
t_j\<=z_{\si(j)}\<+h-l_j\,,\qquad j=m\lc k\,,\quad j\ne i\,,\kern-1.2em
\\[-14pt]
\end{gather*}
while $\,H^{\:\si}_i\<(\TT\:;\zz\:;h)\,$ is a $\,1$-\:periodic function of
\vv.04>
$\,\:t_{m-1}\lc\:t_k\>$ holomorphic in a neighbourhood of the point
$\,\:t_{m-1}\<=z_{\si(i)}\<+h\,$, $\,\:t_j=z_{\si(j)}\<+h\,$, $\,j=m\lc k\,$.
This implies the statement of Lemma \ref{lemFh2}.
\end{proof}

\begin{lem}
\label{lemFh3}
For $\,\:\si\<\in\<S_n[\:k\:]\:$ and nonnegative integers $\,\:l_m\lc l_k\>$,
\vv.08>
$\,m>1\,$, the function $\,F^{\:\si}_{l_m\?\lc\:l_k}\?(\TT^{\les m};\zz\:;h)\>$
is regular at the hyperplanes
$\,\:t_{m-1}\<-z_{\si(i)}\<+l_i\<\in\<\Z_{\le0}\,$, $\,\:i=m\lc k\,$.

\end{lem}
\begin{proof}
By formulae \eqref{Igd}\:, \eqref{FlmRh}\:,
\vvn.4>
\begin{alignat}2
\label{Flmh}
& \;F^{\:\si}_{l_m\?\lc\:l_k}\?(\TT^{\les m};\zz\:;h)\,={}
\\[9pt]
\notag
& {}=\,C^{\:\si}_{l_m\?\lc\:l_k}\?(\zz\:;h)\;
e^{\:\pii\>\sum_{j=1}^{\:m-1}\?t_j}
\,\:W\<(\TT\:;\gmb\:;\gmbb\:;h)|
_{\:t_m=\:z_{\si\<(\?m\?)}\<+\:h-\>l_m\<\lc\>t_k=\:z_{\si\<(\<k\<)}\<+\:h-\>l_k}
^{}\:\times{}
\\[8pt]
\notag
& \:{}\times\>\prod_{i=1}^{m-1}\<\biggl(\:\Gm(z_i\<-t_i)\,\Gm(t_i\<-z_i\<-h)
\,\:\prod_{j=1}^{i-1}\>\frac{(t_i\<-t_j)\,\Gm(t_j\<-t_i\<+h)\,\Gm(z_j\<-t_i)}
{\Gm(1+t_j\<-t_i\<-h)\,\Gm(1+z_j\<-t_i\<+h)}\,\times{}\kern-.4em &&
\\[4pt]
\notag
&& \llap{$\dsize{}\times\>\prod_{j=m}^k\,
\frac{(z_{\si(j)}\?+h-l_j\?-t_i)\,\Gm(t_i\<-z_{\si(j)}\?+l_j)}
{\Gm(1+t_i\<-z_{\si(j)}\?+l_j\?-2\:h)}\:
\prod_{j=i+1}^n\:\frac{\Gm(t_i\<-z_j\<-h)}{\Gm(1+t_i\<-z_j)}
\>\biggr)\>$,}\!\kern-.4em &
\\[-12pt]
\notag
\end{alignat}
where $\,\:C^{\:\si}_{l_m\?\lc\:l_k}\?(\zz\:;h)\,$ does not depend
\vv.1>
on $\,\:\TT\,$. Since $\,\:\si(i)\ge i\,$ for $\,i\le k\,$, the function
$\,F^{\:\si}_{l_m\?\lc\:l_k}\?(\TT^{\les m};\zz\:;h)\,$ has the form
\vvn-.2>
\be
G(\TT\:;\zz\:;h)\,W\<(\TT\:;\gmb\:;\gmbb\:;h)|
_{\:t_m=\:z_{\si\<(\?m\?)}\<+\:h-\>l_m\<\lc\>t_k=\:z_{\si\<(\<k\<)}\<+\:h-\>l_k}
^{}\,\prod_{i=m}^k\,
\frac{\Gm(t_{m-1}\<-z_{\si(i)}\?+l_i)}{\Gm(1+t_{m-1}\<-z_{\si(i)})}\;,
\kern-1.2em
\vv.3>
\ee
where $\,G(\TT\:;\zz\:;h)\,$ as a function of $\,\:t_{m-1}\>$ is holomorphic
\vv.07>
in a neighbourhood of the point $\,\:t_{m-1}\<=z_{\si(i)}\?-l_i\<-r\,\:$
for any $\,\:i=m\lc k\,$, \:and $\>r\?\in\<\Z_{\ge0}\,$. The ratio of
Gamma\:-\:functions
\vvn.4>
\be
\frac{\Gm(t_{m-1}\<-z_{\si(i)}\?+l_i)}{\Gm(1+t_{m-1}\<-z_{\si(i)})}\kern-2em
\vv.4>
\ee
is a polynomial if $\,\:l_i\ge 1\,$, and equals
$\,(t_{m-1}\<-z_{\si(i)})^{-1}\:$ if $\,\:l_i\<=0\,$. In the last case,
\vv.04>
the function $\,F^{\:\si}_{l_m\?\lc\:l_k}\?(\TT^{\les m};\zz\:;h)\,$
does not have a pole at $\,\:t_{m-1}\<=z_{\si(i)}\,$ due to Lemma \ref{Wth}.
Lemma \ref{lemFh3} is proved.
\end{proof}

Recall the integral $\,\:\Ie(\zz\:;h\:;p)\,$, see \eqref{Ie}\:,
and the functions $\,F^{\:\si}_{l_m\?\lc\:l_k}\?(\TT^{\les m};\zz\:;h)\,$,
\vv.07>
see \eqref{FlmRh}\:. Denote by $\,S_n[\:k,m\:]\,$ the set of
$\,\:\si\<\in\<S_n\>$ such that $\,\:\si(i)=i\,\:$ if $\,\:i<m\,$,
$\,\:\si(j)\ge j\,\:$ if $\,\:m\le j\le k\,$, and $\,\:\si(q)<\si(r)\,\:$ if
$\,\:k<q<r\,$. For instance, $\,S_n[\:k,1\:]=S_n[\:k\:]\,$.

\begin{prop}
\label{lemIeResFh}
Let $\,\:|\:p\:|>1\,$. For any $\,m=1\lc k\,$, we have
\vvn.4>
\beq
\label{Iemsh}
\Ie(\zz\:;h\:;p)\,=\!\<
\sum_{l_m\?\lc\:l_k=\:0\!\!}^\infty\;\>\sum_{\si\in S_n\<[\:k,\:m\:]\!\!\!\!}
\,\Ie^{\:\si}_{l_m\?\lc\:l_k}\?(\zz\:;h\:;p)\,
p^{\>\sum_{i=m}^k\<(\<z_{\si\<(\<i\<)}\<+\:h-\:l_i\<)}\:.\kern-2.3em
\vv.4>
\eeq
where $\,\:\Ie^{\:\si}_{l_1\?\lc\:l_k}\?(\zz\:;h\:;p)=
F^{\:\si}_{l_1\?\lc\:l_k}\?(\zz\:;h)\:$ does not depend on $\,p\,$,
and for $\,m>1\,$,
\vvn.4>
\beq
\label{Iemh}
\Ie^{\:\si}_{l_m\?\lc\:l_k}\?(\zz\:;h\:;p)\,=\,
\frac1{\bigl(2\:\piit\>\bigr)^{\<m-1}}
\int\limits_{\Ibt_m\<(\<\zz;h)}\!\!\<p^{\>\sum_{j=m+1}^{\:k}\?t_j}\:
F^{\:\si}_{l_m\?\lc\:l_k}\?(\TT^{\les m};\zz\:;h)\;d\:t_1\:\ldots d\:t_{m-1}\,,
\kern-.66em
\eeq
the integration contour $\,\:\Ibt_{\:m}(\zz\:;h)\,$ being such that
\vv.08>
$\,\:{\Re\:t_i\<=\Re\:(z_i\<+h/2)}\,$ and $\,\,\Im\:t_i\>$ runs
from $\,-\:\infty\,$ to $\,\infty\,$ for all $\,i=1\lc m-1\,$.
\end{prop}
\begin{proof}
The statement is proved by induction on $\,m\,$. The base of induction at
\vv.04>
$\,m=k\,$ amounts to evaluating integral \eqref{Ie} with respect to $\,\:t_k\,$
\vv.06>
using the sum of residues in the half\:-plane
$\,\:{\Re\:t_k\?<\Re\:(z_k\<+h/2)}\,$. The relevant poles of the integrand
are of the form $\,\:t_k\<=z_i\<+h-l_k\>$, $\,\:i\ge k\,$,
$\,\:l_k\?\in\Z_{\ge 0}\,$, and a given $\,\:i\,$ corresponds to
the terms in formula \eqref{Iemsh} labeled by the unique permutation
$\,\:\si\<\in\<S_n[\:k,k\:]\,$ such that $\,\:\si(k)=i\,$.

\vsk.2>
For the induction step, we fix $\,\:\si\<\in\<S_n[\:k,m\:]\,$ and evaluate
\vv.04>
integral \eqref{Iemh} with respect to $\,\:t_{m-1}\,$ using the sum of residues
\vv.04>
in the half\:-plane $\,\:{\Re\:t_{m-1}\<>\Re\:(z_{m-1}\<+h/2)}\,$.
By formula \eqref{Flmh} and Lemma \ref{lemFh3}, the relevant poles of the
\vv.06>
integrand are of the form $\,\:t_{m-1}\<=z_i\?+h-l_{m-1}\>$, $\,i\ge m-1$,
$\,\:l_{m-1}\?\in\Z_{\ge 0}\,$. Therefore,
\vvn.2>
\beq
\label{Iell}
\Ie^{\:\si}_{l_m\?\lc\:l_k}\?(\zz\:;h\:;p)\,=\!
\sum_{l_{m-1}=\:0\!\!}^\infty\,\,\sum_{i=m-1\!}^n\<
\Ie^{\:\si\?,\:i}_{l_{m-1}\?\lc\:l_k}\?(\zz\:;h\:;p)\,
p^{\>z_i+\:h-\:l_{m-1}}\>,\kern-2em
\vv-.4>
\eeq
where
\vvn.3>
\begin{align*}
& \Ie^{\si\?,\:i}_{l_{m-1}\?\lc\:l_k}\?(\zz\:;h\:;p)\,={}
\\[4pt]
& \!\?{}=\,\frac1{\bigl(2\:\piit\>\bigr)^{\<m-2}}
\int\limits_{\Ibt_{m-1}\<(\<\zz;h)}\!\!\<p^{\>\sum_{j=1}^{\:m-2}\?t_j}\:
\Res_{\>t_{m-1}=\:z_i+\:h-\:l_{m-1}}\?
F^{\:\si}_{\:l_m\?\lc\:l_k}\?(\TT^{\les m};\zz\:;h)
\;d\:t_1\:\ldots d\:t_{m-2}\,.\kern-.6em
\end{align*}

\vsk.4>
Consider the bijection between the set $\,S_n[\:k,m-1\:]\,$ and
\vv.04>
the set of pairs $\,(\si,i\:)\,$ such that $\,\:\si\<\in\<S_n[\:k,m\:]\,$,
$\,\:i\in\{\:m-1\lc n\:\}\,$ and $\,\:i\not\in\{\:\si(m)\lc\si(k)\:\}\,$.
\vv.04>
It is given by sending $\,\rho\<\in\<S_n[\:k,m-1\:]\,$ to the pair
$\,\bigl(\si,\rho(m-1)\<\bigr)\,$, such that $\,\:\si\<(m-1)=m-1\,\:$
\vv.06>
and $\,\:\si(i)=\rho(i)\,\:$ if $\,m\le i\le k\,$. We have
\vvn.3>
\beq
\label{FlmRih}
\Res_{\>t_{m-1}=\:z_{\rho\<(\?m-1\?)}\<+\:h-\:l_{m-1}}\?
F^{\:\si}_{\:l_m\?\lc\:l_k}\?(\TT^{\les m};\zz\:;h)\,=\,
F^{\:\rho}_{\:l_{m-1}\?\lc\:l_k}\?(\TT^{\les m-1};\zz\:;h)\,,\kern-2em
\vv.5>
\eeq
Plugging series \eqref{Iell} for
$\,\:\Ie^{\:\si}_{l_m\?\lc\:l_k}\?(\zz\:;h\:;p)\,$ into \eqref{Iemsh}
and using formula \eqref{FlmRih}\:, we get
\vvn.1>
\begin{align}
\label{Iemmh}
\kern-1em\Ie(\zz\:;h\:;p)\,={}&\!\<\sum_{l_{m-1}\?\lc\:l_k=\:0\!\!}^\infty
\;\>\sum_{\rho\:\in S_n\<[\:k,\:m-1\:]\!\!\!\!}\,
\Ie^{\:\rho}_{l_{m-1}\?\lc\:l_k}\?(\zz\:;h\:;p)\,
p^{\>\sum_{j=m-1}^k\<(\<z_{\si\<(j\<)}\<+\:h-\:l_j\<)}\:+{}\kern-2em
\\[4pt]
\notag
{}+{}&\!\<\sum_{l_{m-1}\?\lc\:l_k=\:0\!\!}^\infty\;\>
\sum_{\si\in S_n\<[\:k,\:m\:]\!\!\!\!}\,\,\;\sum_{r=m}^k
\,\:\Ie^{\:\si\?,\:\si(\<r\<)}_{l_{m-1}\?\lc\:l_k}\?(\zz\:;h\:;p)
\,p^{\>z_{\si\<(\<r\?)}\<+\:h-\:l_{m-1}+
\sum_{j=m}^k\<(\<z_{\si\<(j\<)}\<+\:h-\:l_j\<)}\:.\kern-1.7em
\\[-11pt]
\notag
\end{align}
The first summand in the right-hand side comes from the terms in \eqref{Iell}
\vv.06>
with $\,\:i\not\in\{\:\si(m)\lc\si(k)\:\}\,$ and the second one comes from
\vvn.5>
the terms with $\,\:i\in\{\:\si(m)\lc\si(k)\:\}\,$. By Lemma~\ref{lemFh2},
\be
\Ie^{\:\si\?,\:\si(j\<)}_{l_{m-1}\?\lc\:l_k}\?(\zz\:;h\:;p)
\,=\:-\,\Ie^{\:\si\?,\:\si(j\<)}
_{\:l_j\?,\>l_m\?\lc\:l_{j-1}\<,\>l_{m-1}\?,\>l_{j+1}\?\lc\:l_k}\?
(\zz\:;h\:;p)\,.\kern-2em
\vv.5>
\ee
Therefore, the second summand in the right-hand side of formula \eqref{Iemmh}
equals zero. This completes the induction step. Proposition \ref{lemIeResFh}
is proved.
\end{proof}

Formula \eqref{Iemsh} for $\,m=1\,$ coincides with formula \eqref{IeResFh}\:.
\vv.04>
This completes the proof of formula \eqref{IeResFh} and thus the proof of
equality \eqref{Psin8} for $\,\:|\:p\:|>1\,$. Theorem \ref{thmPsin} is proved.

\subsection{Proof of Theorem \ref{trans}}
\label{secpf}
Let $\,\:p\in\CRp\,$.
Given $\,\si\<\in\<S_n\,$, Theorem \ref{thmPsin} implies that
\vvn.4>
\beq
\label{PsiY08}
\Psf_{\<Y^{\<\infty}_{\<\si\vp1}}(\zz\:;h\:;p)\,=\,
\Pso_{\?Y^0_{\<\si}}\:(\zz\:;h\:;p)\kern-1.8em
\vv.3>
\eeq
for any $\,\zz,h\,$ such that $\,(\zz_\si,h)\<\in\Lin\?$.
\vv-.11>
Since both functions $\,\:\Pso_{\?Y^0_{\<\si}}\<(\zz\:;h\:;p)\:$ and
\vv.1>
$\,\:\Psf_{\<Y^{\<\infty}_{\<\si}\vp1}\<(\zz\:;h\:;p)\,$ are holomorphic
as function of $\,\zz,h\,$ for $\,(\zz,h)\in\Lp$, see Theorem \ref{regz=z},
\vv.05>
equality \eqref{PsiY08} holds for any $\,(\zz,h)\in\Lp$ by analytic
\vv.04>
continuation. By formulae \eqref{Yofpi} and definition \eqref{tau} of the
transition map $\,\:\tau\,$, we have $\,\:\Yf_{\<\si}\?=\tau\>\Yo_{\<\si}\:$.
Hence, Theorem \ref{trans} holds for all the classes $\,\:\Yo_{\<\si}\:$.

\vsk.3>
For any $\,X\?\in\<\KT\,$, formula \eqref{WofX} and Proposition \ref{piWX}
imply that
\vvn.4>
\be
X\,=\<\sum_{I\in\>\Ikn\!}
\frac{X|_I^{}(\ZZ;H)}{R(\ZZ_{\?\si_I}\<)}\;\Yo_{\<\si_I}\kern-2em
\vv.2>
\ee
in the extension of $\,\KT\,$ by rational functions of $\,\ZZ\>$. Thus
\vvn.5>
\be
\tau X\,=\<\sum_{I\in\>\Ikn\!}
\frac{X|_I^{}(\ZZ;H)}{R(\ZZ_{\?\si_I}\<)}\;\Yf_{\<\si_I},\kern-2em
\vv.1>
\ee
and formula \eqref{PsiY08} yields
$\,\:\Psf_{\?\tau X}(\zz\:;h\:;p)=\Pso_{\?X}(\zz\:;h\:;p)\,$
\vv.07>
for any $\,X\?\in\<\KT\,$. Theorem~\ref{trans} is proved.

\section{Proofs}
\label{sec:pfs}

\subsection{Proof of Propositions \ref{p<>1}}
\label{sec:props}
It is known that for any compact subset $\,K\,$ of
$\,(\:\C\?\setminus\?\Z_{\le0})\<\times\<(\:\C\?\setminus\?\Z_{\le0})\,$,
there is a constant $\,C\>$ such that
\vvn.3>
\beq
\label{Stirk}
\biggl|\:\frac{\Gm(\al+k)}{\Gm(\bt+k)}\:\biggr|\,\le\,C\:(k+1)^{\>\al-\bt}
\kern-2em
\vv.4>
\eeq
for any $\,k\in\Z_{\ge0}\,$ and any $\,(\al,\bt)\in\<K\,$.
Recall also the formula
\vvn.4>
\beq
\label{Eur}
\Gm(x)\>\Gm(1-x)\,=\,\frac\pi{\sin\:(\pi x)}\;.\kern-2em
\vv.4>
\eeq
Then by formulae \eqref{Phr}\:, \eqref{Eur}\:, \eqref{Stirk}\:,
for any $\,I\?\in\Ikn\,$ and any compact subset $\,K\,$ of $\,L\,$,
there is a constant $\,C\>$ such that
\vvn.5>
\beq
\label{estim}
\bigl|\>\Res_{\>\TT\:=\zz_I\<+\:\lb}\Phr\<(\TT\:;\zz\:;h)\bigr|\,\le\,
C\>\bigl(k+\<{\tsize\sum_{\:i=1}^{\:k}}\:l_i\:\bigr)^{k-1-k\:h(n+1-k)}
\kern-1em
\vv.4>
\eeq
for any $\,\:\lb\<\in\Z_{\ge 0}^k\,$ and any $\,(\zz,h)\in\<K\,$.
\vv.04>
This estimate implies similar estimates for the derivatives because
$\,\:\Res_{\>\TT\:=\zz_I\<+\:\lb}\bigl(\Phr\<(\TT\:;\zz\:;h)\<\bigr)\,$
\vv0>
is a holomorphic function of $\,\zz,h\>$ for $\,(\zz,h)\in L\,$. Since
\be
\Res_{\>\TT\:=\zz_I\<+\:\lb}
\bigl(\Phr\<(\TT\:;\zz\:;h)\>W\<(\TT\:;\gmb\:;\gmbb\:;h)\<\bigr)\,=\,
W\<(\zz_I\<+\:\lb\:;\gmb\:;\gmbb\:;h)\>
\Res_{\>\TT\:=\zz_I\<+\:\lb}\Phr\<(\TT\:;\zz\:;h)\kern-.4em
\vv.6>
\ee
and $\,\:W\<(\TT\:;\gmb\:;\gmbb\:;h)\,$ is a polynomial, estimate \eqref{estim}
\vv.06>
shows that the series in the right\:-hand side of \eqref{Fo} converges
\vv.04>
to a function holomorphic in $\,\:p\:$ for $\,|\:p\:|\<<1\,$ and holomorphic
in $\,\zz,h\>$ for $\,(\zz,h)\in L\,$. This proves Proposition \ref{p<>1}
for the function $\,\Fo_{\?I}(\zz\:;h\:;p)\,$.

\vsk.4>
The proof of Proposition \ref{p<>1} for the function
$\,\Ff_{\?I}(\zz\:;h\:;p)\,$ is similar.

\subsection{Determinant formula}
\label{sec:thms}
Recall
\vvn.1>
$\,\Hc_{\zz_0,\:h_0}\?=\HT\big/\bra\:\zz=\zz_0\,,\,h=h_0\:\ket\,$.
Consider the polynomials $\,V_{\<I}\,$, $\,I\?\in\Ikn\,$,
defined by formula \eqref{Sch}\:.
\begin{lem}
\label{VH}
For any $\,(\zz_0\:,h_0)\,$, the classes $\,\:V_{\<I}(\gmb)\,$,
$\>I\?\in\Ikn\,$, \>form a basis of $\>\Hc_{\zz_0,\:h_0}$.
\end{lem}
\begin{proof}
The statement follows from formula \eqref{Sch2} and Corollary \ref{lemfV}.
\end{proof}

Recall $\,\:\Ub\,$ the universal cover of $\,\:\Czon\,$ and $\,\:\BB\,$
\vv.06>
the image of $\,\:\CRp\,$ in $\,\:\Ub\,$, see Section \ref{secmono}.
The functions $\,p^{\:s}\<$ and $(1-p)^s$ are defined on $\,\:\Ub\,$
such that their values on $\,\:\BB\,$ are determined by the inequalities
$\,\:-\:2\:\pi\<<\<\arg\:p<0\,$ and $\,|\<\arg\:(1-p)\:|<\pi\,$, respectively.
\vsk.4>
For $\,(\zz_0\:,h_0)\in\Lp$, $\,p_0\?\in\Ub$, and the classes
$\,\:V_{\<I}(\GG)\in\KT\,$, $\,I\?\in\Ikn\,$,
\vv.06>
expand the elements $\,\:\Pso_{\?V_{\?I}\<(\<\GG\<)}\?(\zz_0\:;h_0\:;p_0)$
\,\:and $\,\:\Psf_{\?V_{\?I}\<(\<\GG\<)}\?(\zz_0\:;h_0\:;p_0)$ using the basis
\vvn.3>
$\,\:V_{\!J}(\gmb)\,$, $\>J\?\in\Ikn\,$, of $\,\Hc_{\zz_0,\:h_0}\,$,
\begin{gather*}
\Pso_{\?V_{\?I}\<(\<\GG\<)}\?(\zz_0\:;h_0\:;p_0)\,=\<
\sum_{J\in\>\Ikn\!}\?\Mo_{\IJ}(\zz_0\:;h_0\:;p_0)\>V_{\!J}(\gmb)\,,\kern-2em
\\[6pt]
\Psf_{\?V_{\?I}\<(\<\GG\<)}\?(\zz_0\:;h_0\:;p_0)\,=\<
\sum_{J\in\>\Ikn\!}\?\Mf_{\IJ}(\zz_0\:;h_0\:;p_0)\>V_{\!J}(\gmb)\,.\kern-2em
\end{gather*}

\begin{prop}
\label{lemdetM}
Let $\,(\zz,h)\in\Lp\!$. Then
\vvn.2>
\begin{align}
\label{detMo}
\det\:\bigl(\:\Mo_{\IJ}(\zz\:;h\:;p)\bigr)_{\<\IJ\in\>\Ikn}\?
&{}=\,\bigl(\:2\:\piit\>\bigr)
^{\tsize\!\<\frac{n\:(n-1)}2\<\binom{\:n-2\:}{k-1}}_{\vp1}\:
p^{\:k\binom{\:n\:}k_{\vp1}\?\sum_{a=1}^n\?z_a}\:
(1-p)^{\:hk\<\binom{\:n\:}k_{\vp1}}\times{}\!\!\kern-1.8em
\\[3pt]
\notag
&{}\>\times\,
e^{\:3\pii\,\:k\binom{\:n-1\:}{k-1}_{\vp1}\?\sum_{a=1}^n\?z_a}\>
\prod_{a=1}^n\>\prod_{\satop{b\:=1}{b\ne a}}^n\,
\Gm(z_a\<-z_b\<-h)^{\binom{\:n-2\:}{k-1}_{\vp1}}\kern-1.8em
\\[-29pt]
\notag
\end{align}
and
\vvn.3>
\beq
\label{detMf}
\det\:\bigl(\:\Mf_{\IJ}(\zz\:;h\:;p)\bigr)_{\<\IJ\in\>\Ikn}\?=\,
e^{\:\pii\,\:hk(n-1)\<\binom{\:n\:}k_{\vp1}}\>
\det\:\bigl(\:\Mo_{\IJ}(\zz\:;h\:;p)\bigr)
_{\<\IJ\in\>\Ikn}\>.\kern-1.5em
\vv.3>
\eeq
\end{prop}
\begin{proof}
Formula \eqref{detMo} follows from \cite[Theorem~4.10\:]{TV5} and
Lemma \ref{lemdetV}. Formula \eqref{detMf} follows from Theorem \ref{trans}
and Proposition \ref{dettau}.
\end{proof}

\appendix

\section{Schur polynomials}
\label{appSch}
Given a subset $\,I=\{\:i_1\?\lsym<i_m\}\,$ of $\,\:\onen\,$,
define a polynomial
\vvn,4>
\beq
\label{Sch}
V_{\<I}(x_1\lc x_m)\,=\,\frac{\det\:(x_{\<a}^{\:i_b}\<)_{\ab=1}^{\:m}}
{\det\:(x_{\<a}^{\:b}\:)_{\ab=1}^{\:m}}\;,\kern-2em
\vv.3>
\eeq
symmetric in $\,x_1\lc x_m\>$. Up to a change of notation, the polynomials
$\,V_{\<I}(x_1\lc x_m)\,$ coincide with the Schur polynomials.

\vsk.3>
For $\,I=\{\:i_1\?\lsym<i_m\}\<\subset\onen\,$, set
$\;\ell(I)=\sum_{\:a=1}^{\:m}\:(i_a\<-a)\,$ and
$\,\:\xx_{\<I}^{}=(x_{i_1}\<\lc x_{i_m})\,$. Let $\,\si_0\,$
be the longest permutation, $\,\si_0(a)=n+1-a\,$, $\,a=1\lc n\,$.

\begin{lem}
\label{lemdetV}
We have,
\vvn-.7>
\beq
\label{detV}
\det\:\bigl(V_{\<I}(\xx_J)\bigr)_{\IJ\in\Ikn}\?=\,\prod_{a=1}^{n-1}\,
\prod_{b=a+1}^n(x_b-x_a)^{\?\binom{n-2}{k-1}}_{\vp x}\:.\kern-1em
\vv.1>
\eeq
\end{lem}
\begin{proof}
Let $\,A\,$ and $\,B\,$ be the left\:-hand and right\:-hand sides of formula
\eqref{detV}\:, respectively. They are polynomials in $\,\xxx\,$. Since
the homogeneous degree of $\,V_{\<I}(x_1\lc x_k)\,$ equals $\,\ell(I)\,$,
see \eqref{Sch}\:, the homogeneous degree of $\,A\,$ is
\vvn.2>
\be
\sum_{I\in\>\Ikn}\!\ell(I)\,=\,
\frac12\sum_{I\in\>\Ikn}\!\bigl(\ell(I)+\ell(\si_0(I)\<)\<\bigr)\:=\>
\frac{k\>(n-k)}2\binom{\:n\:}k=\>\frac{n\>(n-1)}2\>\binom{n-2}{k-1}\>,
\kern-.4em
\vv.1>
\ee
that equals the homogeneous degree of $\,B\,$. Furthermore, $\,A\,$ is divisible
by $\,B\,$ by the standard reasoning. Therefore, $\,A\,$ and $\,B\,$
\vv.1>
are proportional. Comparing the coefficients of the monomial
$\,\prod_{\:a=2}^{\:n}x_a^{(a-1)\smash{\binom{n-2}{k-1}}}\:$
in $\,A\,$ and $\,B\,$ completes the proof.
\end{proof}

For $\,I\subset\onen\,$, let $\,\Ibr=\onen\setminus I\,$ be the complement of
$\,I\,$. Set
\vvn.2>
\beq
\label{Rx}
R(x_1\lc x_m\:;y_1\lc y_{n-m})\,=\,
\prod_{a=1}^m\,\prod_{b=1}^{n-m}\,(y_{\:b}-x_a)\,.\kern-2em
\eeq

\vsk.3>
The Laplace expansion of the determinant
\vv.1>
$\;\det\:(x_{\<a}^{\:b}\:)_{\ab=1}^{\:n}\,$ with respect of the first $\,k\,$
and the last $\,n-k\,$ rows together with the Vandermonde determinant formula
yield
\vvn.6>
\beq
\label{Sch1}
\sum_{I\in\>\Ikn}\?(-1)^{\ell(I)}\>
V_{\<I}(x_1\lc x_k)\,V_{\<\Ibr}(x_{k+1}\lc x_n)\,=\,
R(x_1\lc x_k\:;x_{k+1}\lc x_n)\,,\kern-1em
\eeq
and the Laplace expansion of $\;\det\:(x_{\<a}^{\:b}\:)_{\ab=1}^{\:n}\,$ with
respect of the first $\,k\,$ and the last $\,n-k\,$ columns yields
\vvn-.1>
\beq
\label{Sch2}
\sum_{I\in\>\Ikn}\!\frac{V_{\?J}(\xx_I)\,V_{\<\Kbr}(\xx_{\Ibr})}
{R(\xx_I\:;\xx_{\Ibr})}\;=\,(-1)^{\ell(J)}\>\dl_{\JK}\,.\kern-2em
\vv.4>
\eeq

\begin{cor}
\label{lemfV}
For any $\,f\in\C[x_1^{\pm1}\?\lc x_n^{\pm1}\:]^{\:S_k\times\:S_{n-k}}_{\vp1}$,
we have
\vvn.3>
\beq
\label{fV}
f(\xxx)\,=\sum_{I\in\>\Ikn}\?(-1)^{\ell(I)}\>V_{\<I}(x_1\lc x_k)\<
\sum_{J\in\Ikn}\!\frac{f(\xx_J,\xx_{\Jbr})\,V_{\<\Ibr}(\xx_{\Jbr})}
{R(\xx_J\:;\xx_{\Jbr})}\;.\kern-1em
\vv.2>
\eeq
\end{cor}
\begin{proof}
Changing the order of summation in the right\:-hand side of \eqref{fV} and
applying formulae \eqref{Sch2}\:, \eqref{Sch1}\:, \eqref{Rx}\:, one has
\vvn.5>
\begin{align*}
\sum_{J\in\Ikn}\frac{f(\xx_J,\xx_{\Jbr})}{R(\xx_J\:;\xx_{\Jbr})}\,
&\sum_{I\in\>\Ikn}\?(-1)^{\ell(I)}\>V_{\<I}(x_1\lc x_k)\,V_{\<\Ibr}(\xx_{\Jbr})
\,={}
\\[4pt]
{}={}&\sum_{J\in\Ikn}\!\frac{f(\xx_J,\xx_{\Jbr})\,R(x_1\lc x_k\:;\xx_{\Jbr})}
{R(\xx_J\:;\xx_{\Jbr})}\;=\,f(\xxx)\,,
\\[-15pt]
\end{align*}
since $\,R(x_1\lc x_k\:;\xx_{\Jbr})=0\,\:$ unless $\,J=\onek\,$,
$\,\Jbr=\{\:k+1\lc n\:\}\,$.
\end{proof}

\end{document}